\documentclass{article}
\usepackage{amssymb,bm}
\usepackage{amsmath}
\usepackage{amsthm}
\usepackage{amsfonts}
\usepackage{mathrsfs}
\usepackage{appendix}
\usepackage{hyperref}
\usepackage{authblk}
\usepackage{color}

\usepackage{mathtools}

\DeclarePairedDelimiter\floor{\lfloor}{\rfloor}

\newtheorem{theorem}{Theorem}[section]
\newtheorem{lemma}[theorem]{Lemma}

\theoremstyle{remark}

\usepackage{graphics}
\usepackage{graphicx}
\usepackage{geometry}
\usepackage{algorithm}
\usepackage{algorithmic}
\hypersetup{colorlinks,linkcolor={blue},citecolor={blue},urlcolor={red}}

\usepackage[english]{babel}

\usepackage{cite}
\usepackage{threeparttable}
\usepackage{float}
\floatstyle{plaintop}
\restylefloat{table}

\floatstyle{plain}
\restylefloat{figure}
\usepackage{placeins}

\begin{document}

\title{How did Donald Trump Surprisingly Win the 2016 United States Presidential Election? an Information-Theoretic Perspective (i.e. Clean Sensing for Big Data Analytics: Optimal Sensing Strategies and New Lower Bounds on the Mean-Squared Error of Parameter Estimators which can be Tighter than the Cram\'{e}r-Rao Bound)}

\author{
Weiyu Xu {\footnote{
Department of Electrical and Computer Engineering, University of Iowa, Iowa City, IA 52242. Corresponding email: weiyu-xu@uiowa.edu.}}\,~~~~~~~
Lifeng Lai {\footnote{
Department of Electrical and Computer Engineering, University of California, Davis, CA, 95616.}}\,~~~~~~~
Amin Khajehnejad {\footnote{
3Red Trading Group, LLC, Chicago, IL. }}\,

}

\maketitle
\begin{abstract}
Donald Trump was lagging behind in nearly all opinion polls leading up to the 2016 United States presidential election of Tuesday, November 8, 2016,  but Donald Trump surprisingly won the presidential election.  Due to  the significance of the United States presidential elections,  this raises the following important questions: 1) why most opinion polls were not accurate in 2016?  and 2) how to improve the accuracies of opinion polls?   In this paper,  we study and explain the inaccuracies of opinion polls in the presidential election of 20016 through the lens of information theory. We first propose a general  framework of parameter estimation in information science, called clean sensing (polling),   which performs optimal parameter estimation with sensing cost constraints,  from heterogeneous and potentially distorted data sources. We then cast the opinion polling as a  problem of parameter estimation from potentially distorted heterogeneous data sources, and derive the optimal polling strategy using heterogenous and possibly distorted data under cost constraints.   Our results show that a larger number of data samples do not necessarily lead to better polling accuracy, which give a possible explanation of the inaccuracies of most opinion polls for the 2016 presidential election. The optimal sensing (polling) strategy should instead optimally allocate sensing resources over heterogenous data sources according to several factors including data quality, and, moreover, for a particular data source, the optimal sensing strategy should strike an optimal balance between the quality of data samples, and the quantity of data samples.

As a byproduct of this research, in a general setting beyond the clean sensing problem,  we derive a group of new lower bounds on the mean-squared errors of general unbiased and biased parameter estimators. These new lower bounds can be tighter than the classical Cram\'{e}r-Rao bound (CRB) and Chapman-Robbins bound. Our derivations are via studying the Lagrange dual problems of certain convex programs. The classical Cram\'{e}r-Rao bound and Chapman-Robbins bound follow naturally from our results for special cases of these convex programs.
\end{abstract}

Keywords: parameter optimization, the Cram\'{e}r-Rao bound (CRB), the Chapman-Robbins bound, information theory, polling, heterogeneous data.

\section{Introduction}

     In many areas of science and engineering, we are faced with the task of estimating or inferring certain parameters from heterogenous data sources. These heterogenous data sources can have data of different qualities for the task of estimation: the data from certain data sources can be more noisy or more distorted than those from other data sources.  For example, in sensor networks monitoring trajectories of moving targets,  the sensing data from different sensors can have different signal to noise ratios, depending on factors such as distances between the moving target and the sensors, and precisions of sensors. As another example,  in political polling, the polling data can come from diverse demographic groups, and the polling data from different demographic groups can have different levels of noises and distortions for a particular qusestionaire.

     Even from within a single data source, one can also obtain data of different qualities for inference, through different sensing modalities of different costs.  For example, when we use sensors with higher precisions to sense data from a given data source, we can obtain data of higher quality, but at a higher sensing cost.  Moreover,  when we try to estimate the parameters of interest,  we often operate under sensing cost constraints, namely the total costs spent on obtaining data from heterogenous data sources cannot be above a certain threshold.

     This raises the natural question: ``Under given cost constraints, how do we perform optimal estimation of the parameters of interest, from heterogenous data?''  By ``optimal estimation'', we mean minimizing the estimation error in terms of certain performance metrics, such as minimizing the mean squared error.

    In this paper, we propose a generic framework to answer the question above, namely to optimally estimate the parameters of interest from heterogenous data, under certain cost constraints. In particular,  we consider how to optimally allocate sensing resources to obtain data of heterogeneous qualities from heterogeneous data sources, to achieve the highest fidelity in parameter estimation.

    Our research is partially motivated by the actual results of the 2016 United States presidential election of Tuesday, November 8, 2016, and the polling results before the election which are mostly contradictory to the actual election results. Donald Trump was lagging behind in nearly all opinion polls leading up to the 2016 United States presidential election, but Donald Trump surprisingly won the presidential election with his 306 electoral votes (state-by-state tallies, without accounting for faithless electors) versus Hillary Clinton's 232 electoral votes (state-by-state tallies, without accounting for faithless electors). Right before the election in 2016, some polling analysts were very confident about the prediction that Hillary Clinton would win the US presidency: there was nearly a unanimity among forecasters in predicting a Clinton victory.  A notable example is that, neuroscientist and polling analyst Sam Wang, one of the founders of Princeton Election Consortium, predicted that a greater than 99$\%$ chance of a Clinton victory in his Bayesian model \cite{2016predictionwang1,2016predictionwang2},  as seen in Wang's election morning blog post titled ``Final Projections: Clinton 323 EV, 51 Democratic Senate seats, GOP House'' \cite{grading, mostclinton}. As an anecdote, before the election, being very confident with his predictions that Hillary Clinton would win the election, Dr. Wang made a promise to eat a bug if Donald J. Trump won more than 240 Electoral College votes, which he later kept by eating a cricket with honey on CNN \cite{citmagazine}.

The contrast between most poll predictions and the actual results of the 2016 presidential election was so dramatic that it was surprising and puzzling to many pollsters. In fact, the actual election results differed from the polling results evidently, sometimes dramatically, both nationally and statewise.  Donald Trump performed better in the fiercely competitive battlegroup Midwestern states where the polls predicted Trump had an advantage, such as Iowa, Ohio, and Missouri, than expected. Trump also won Wisconsin, Michigan and Pennsylvania, which were considered part of the blue firewall.  For example, let us consider the final polling average published by Real Clear Politics on November 7, 2016. The poll average showed that, in Wisconsin, Clinton had a +6.5$\%$ advantage over Trump, while the actual election result showed that Trump had a +0.7$\%$ advantage over Clinton; the poll average showed that, in Michigan, Clinton had a +3.4$\%$ advantage over Trump, while the actual election showed Trump had a +0.3$\%$ advantage over Clinton;  the poll average showed that, in Pennsylvania,  Clinton had a +1.9$\%$ advantage over Trump, while the actual election result had Trump +0.3$\%$ on top.  In Iowa,  the poll average showed that Trump had a +3.9$\%$ advantage over Clinton, while the actual election result showed that Trump's advantage greatly increased to  +9.5$\%$; in Missouri, the poll average showed that Trump had a +9.5$\%$ advantage over Clinton, while the actual election result showed that Trump's advantage greatly increased to  +18.5$\%$; and in Minnesota, Clinton had a +6.2$\%$ advantage over Trump, while the actual election result showed that Clinton's advantage significantly shrunk to  +1.5$\%$. While, in most of the states where the polling results are evidently different from the actual voting results,  Trump outperformed the polling results, Clinton also outperformed the polling results in a small number of states, such as in the states of Nevada, Colorado, and New Mexico.  Figure 1, cited from \cite{pollingdifference},  shows the difference between the final polling average published by Real Clear Politics \cite{realclearpolitics} on November the 7th, and the final voting results in 16 states. It is worth mentioning, among dozens of polls,  the UPI/CVoter poll and the University of Southern California/Los Angeles Times poll were the only two polls that often predicted a Trump popular vote victory or showed a nearly tied election.

\begin{figure}[!t]
\centering
\includegraphics[width =4.5 in]{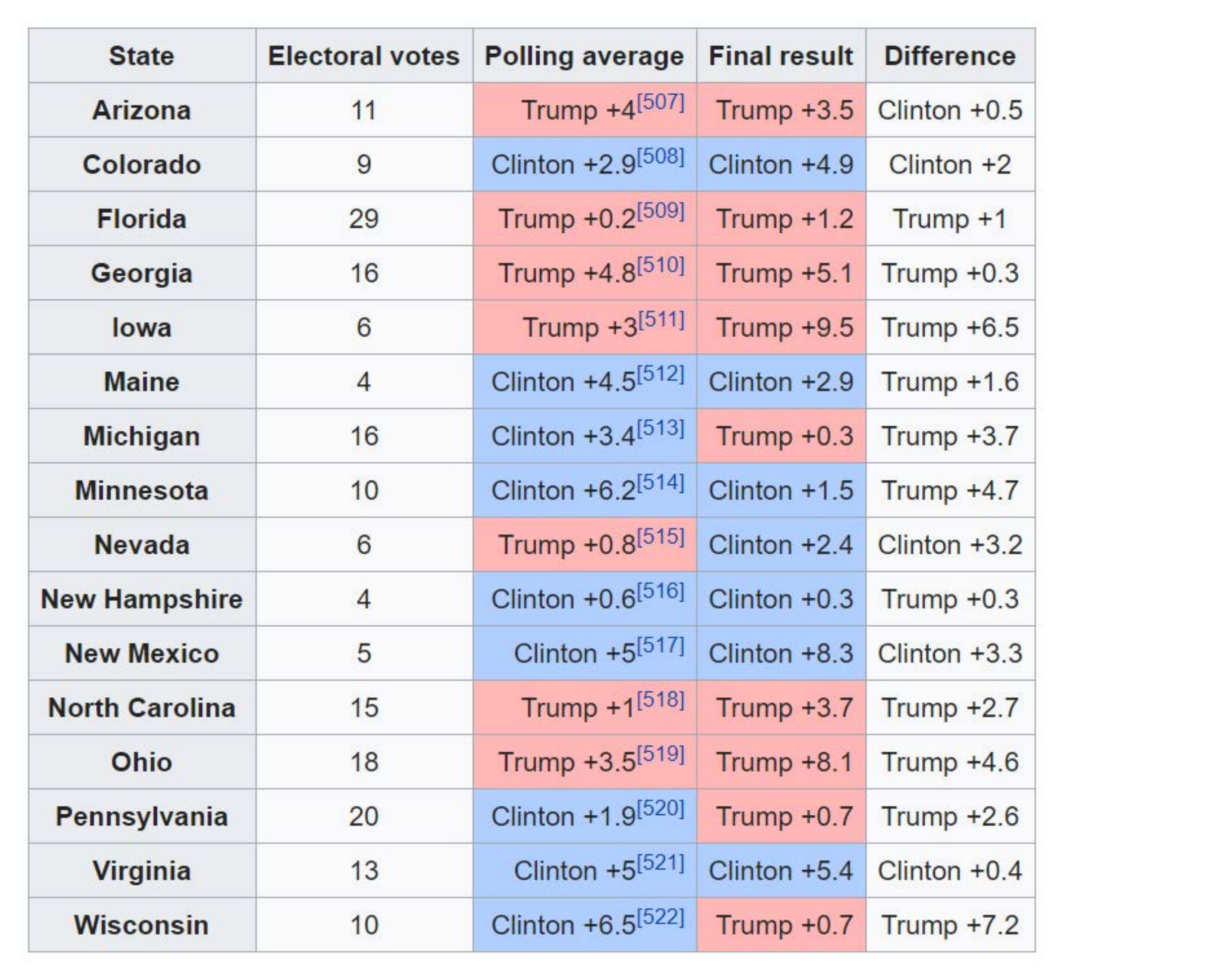}
\caption{Comparisons between the polling average published by Real Clear Politics \cite{realclearpolitics} on November the 7th, and the final voting results in 16 states. This table is cited from \cite{pollingdifference}.} \label{polling}
\end{figure}

     The dramatic and consistent differences between the polling results and  the actual election returns, both nationally and statewise, cannot be explained by the ``margin of errors'' of these polling results. This indicates that there are significant and systematic errors in the polling results. This contrast was also alarming, considering election predictions had already had access to big data, and had applied advanced big data analytics techniques. So it is imperative to understand why the predictions from polling were terribly off.

 Due to  the significance of the United States presidential elections,  this raises the following important questions: 1) why most opinion polls were not accurate in 2016?  and 2) how to improve the accuracies of opinion polls?  While there are many possible explanations for the inaccuracies of the opinion polls for the 2016 presidential election,  in this paper, we look at the possibility that the collected opinion data in polling were distorted and noisy, and heterogeneous  in noises and distortions, across different demographic groups. For example, supporters for a candidate might be embarrassed to tell the truth, and thus more likely to lie in polling, when their friends and/or  local/national news media are vocal supporters for the opposite candidate.

    In this paper,  we study and explain the inaccuracy of most opinion polls through the lens of information theory. We first propose a general  framework of parameter estimation in information science, called clean sensing (polling),   which performs optimal parameter estimation with sensing cost constraints,  from heterogeneous and potentially distorted data sources.   We then cast the opinion polling as a  problem of parameter estimation from potentially distorted heterogeneous data sources, and derive the optimal polling strategy using heterogenous and possibly distorted data under cost constraints.   Our results show that a larger number of data samples do not necessarily lead to better polling accuracy. The optimal sensing (polling) strategy instead optimally allocates sensing resources over heterogenous data sources, and, moreover, for a particular data source, the optimal sensing strategy should strike an optimal balance between the quality of data samples, and the number of data samples.

As a byproduct of this research, we derive a series of new lower bounds on the mean squared errors of unbiased and biased parameter estimators, in the general setting of parameter estimations. These new lower bounds can be tighter than the classical Cram\'{e}r-Rao bound (CRB) and Chapman-Robbins bound. Our derivations are via studying the Lagrange dual problems of certain convex programs, and the classical Cram\'{e}r-Rao bound (CRB) and Chapman-Robbins bound follow naturally from our results for special cases of these convex programs.

   The rest of this paper is organized as follows. In Section \ref{sec:problemformulation}, we introduce the problem formulation of parameter estimation using potentially distorted data from heterogeneous data sources, namely the problem of clean sensing. In Section \ref{sec:CS},  we cast finding the optimal sensing strategies in parameter estimation using heterogeneous data sources as explicit mathematical optimization problems.  In Section \ref{sec:CSsolver},  we derive asymptotically optimal solutions to the optimization problems of finding the optimal sensing strategies for clean sensing.   In Section \ref{sec:Gaussian}, we consider clean sensing for the special case of Gaussian random variables.  In Section \ref{sec:polling}, we cast the problem of opinion polling in political elections as a problem of clean sensing, and give a possible explanation for why the polling for the 2016 presidential election were not accurate. We also derive the optimal polling strategies under potentially distorted data to achieve the smallest polling error.  In Section \ref{sec:lowerbound}, we derive new lower bounds on the mean-squared errors of parameter estimators, which can be tighter than the classical Cram\'{e}r-Rao bound and Chapman-Robbins bound. Our derivations are via solving the Lagrange dual problems of certain convex programs, and are of independent interest.

\section{Problem Formulation}
\label{sec:problemformulation}

In this section, we introduce the problem formulation, and model setup for clean sensing (polling). Suppose we want to estimate a parameter $\theta$ (which can be a scalar or a vector), or a function of the parameter, say, $f(\theta)$.  We assume that there are $K$ heterogeneous data sources,  where $K$ is a positive integer.  From each of heterogeneous data sources, say, the $k$-th data source, we obtain $m_{{k}}$ samples, where $1\leq k \leq K$.   We denote the $m_{k}$ samples from the $k$-th data source as $X_{1}^{k}$,  $X_{2}^{k}$, ..., and $X_{{m_{k}}}^{{k}}$, and these samples take values from domain $\mathcal{D}_{k}$.  We assume that cost $c^{k}_{i}$ was spent on acquiring the $i$-th sample from the $k$-th data source, where $1\leq k \leq K$, and $1\leq i \leq m_{k}$. We assume that we take action $\mathcal{A}$ in sampling from the $K$ heterogenous data sources.  We assume that under action $\mathcal{A}$,  the $m=\sum_{k=1}^{K} m_{k}$ samples $X_{1}^{1}$,  $X_{2}^{1}$, ..., , $X_{{m_{1}}}^{{1}}$,  $X_{1}^{2}$,  $X_{2}^{2}$, ..., , $X_{{m_{2}}}^{{2}}$, ..., $X_{1}^{K}$,  $X_{2}^{K}$, ..., $X_{{m_{K}}}^{{K}}$ follow  distribution $f( X_{1}^{1},  X_{2}^{1}, ..., , X_{{m_{1}}}^{{1}},  X_{1}^{2},  X_{2}^{2}, ..., , X_{{m_{2}}}^{{2}}, ..., X_{1}^{K},  X_{2}^{K}, ..., , X_{{m_{K}}}^{{K}}, \theta, \mathcal{A} )$. This distribution depends on the parameter $\theta$, and the action $\mathcal{A}$.

 In this paper,  without loss of generality,  we assume that under action $\mathcal{A}$, the samples across data sources are independent, and samples from a single data source are independent.  Hence we can  express the distribution as follows:
 \begin{align*}
 &f( X_{1}^{1},  X_{2}^{1}, ..., , X_{{m_{1}}}^{{1}},  X_{1}^{2},  X_{2}^{2}, ..., , X_{{m_{2}}}^{{2}}, ..., X_{1}^{K}, ..., , X_{{m_{K}}}^{{K}}, \theta, \mathcal{A} )\\
 &=f_{1}^{{1}}(   X_{1}^{1}, \theta, \mathcal{A})f_{2}^{{1}} (X^{1}_{2}, \theta, \mathcal{A}) \cdots f^{1}_{m_{1}} (X^{1}_{m_{1}}, \theta, \mathcal{A}) f_{1}^{{2}}(   X_{1}^{2}, \theta, \mathcal{A})\cdots f_{m_2}^{{2}}(   X_{m_2}^{2}, \theta, \mathcal{A}) \cdots f_{m_{K}}^{K}(X_{{m_{K}}}^{{K}}, \theta, \mathcal{A}),
\end{align*}
where  $f_{i}^{k} (X_{i}^{k})$ is the probability distribution of  $X_{i}^{k}$, namely the $i$-th sample from the $k$-the data source, with $1\leq k \leq K$ and $1\leq i \leq m_{k}$. We can of course also extend the analysis to more general cases where the samples are not independent. 

In this paper, we consider the following problem: under a budget on the total cost for sensing (polling), what is the optimal action $\mathcal{A}$ to guarantee the most accurate estimation of $\theta$ or its function?  More specifically, determining the sampling action means determining the number of data samples from each data source, and determining the cost spent on obtaining each data sample from each data source. We consider the non-sequential setting, where the sampling action is predetermined before the actual sampling action happens.  In this paper, we use the mean-squared error to quantify the accuracy of the estimation of $\theta$ or a function of $\theta$.  We can also extend this work to use other performance metrics than the mean-squared error, such as those concerning the distribution of the estimation error or the tail bound on the estimation error.

\section{Related Works}
In \cite{venutestingcost}, the authors considered controlled sensing for sequential hypothesis testing with the freedom of selecting sensing actions with different costs for best detection performance. Compared with \cite{venutestingcost}, our work is different in three aspects: 1) in \cite{venutestingcost}, the authors were considering a sequential hypothesis testing problem, while in this paper of ours, we consider a parameter estimation problem; 2) in\cite{venutestingcost}, the authors worked with samples from a single data source (possibly having different distributions under different sampling actions), while in our paper, we consider heterogenous data sources where samples follow distributions determined not only by the sampling actions but also by the types of data sources; 3) in our paper, we consider continuously-valued sampling actions whose costs can take continuous values, compared with discretely-valued sampling actions with discretely-valued costs in \cite{venutestingcost}.

In \cite{StatisiticsPaper}, the authors considered designs of experiments for sequential estimation of a function of several parameters, using data from heterogeneous sources (types of experiments), with a budget constraint on the total number of samples. In \cite{StatisiticsPaper}, each data sample from each data source always requires a unit cost to obtain, and the observer does not have the freedom of controlling the data quality of any individual sample.  Compared with \cite{StatisiticsPaper}, in this paper, each data sample can require a different or variable cost to obtain, depending on the type of data source involved, and the specific sampling action used to obtain that data sample. Moreover, in this paper, the quality of each data sample depends both on the type of data source and on the sampling action used to obtain that data sample. For example, in this paper, we  have the freedom of not only optimizing the number of data samples for each data source, but also optimizing the effort (cost) spent on obtaining a particular data sample from a particular data source (depending on the cost-quality tradeoff of that data source); while in \cite{StatisiticsPaper}, one only has the freedom of choosing the number of data samples from each data source.
In \cite{StatisiticsPaper}, each data source (type of experiment) reveals information about only one element of the parameter vector $\theta$; while in this work, a sample from a data source can possibly reveal information about several elements of the parameter vector $\theta$.

\section{Clean Sensing: Optimal Estimation Using Heterogenous Data under Cost Constraints}
\label{sec:CS}
In this section, we introduce the framework of clean sensing, namely optimal estimation using heterogenous data under cost constraints. As explained in Section \ref{sec:problemformulation}, we assume that we spend cost $c_i^k$ on acquiring the $i$-th sample from the $k$-th data source, and,  given the costs spent on acquiring each data sample,  all the samples are independent of each other. Our goal is to optimally allocate the sensing resources to each data source and each data sample, in order to minimize the Cram\'{e}r-Rao bound on the mean-squared error of parameter estimations. One can also extend this framework to minimize other types of bounds on the mean-squared error of parameter estimation.

We consider a parameter column vector denoted by
$$ {  {\theta }}=\left[\theta _{1},\theta _{2},\dots ,\theta _{d}\right]^{T}\in {\mathbb  {R}}^{d}.$$
Under the parameter vector $ f(x;{ {\theta }})$, we assume that probability density function of an observation sample is given by $ f(x;{{\theta }})$. Let $ { {g}}(X)$ be an estimator of any vector function of parameters, $ {  {g}}(X)=(g_{1}(X),\ldots ,g_{d}(X))^{T}$ (note that we can also consider cases where the $g(X)$ be of a different dimension), and denote its expectation vector ${{E} [{ {g}}(X)]}$  by $ {  {\phi }}({  {\theta }})$.

The Fisher information matrix is a $ d \times d$ matrix with its element $ I_{{i,j}}$ in the $i$-th row and $j$-th column defined as  
$$ {\displaystyle I_{i,j}=\operatorname {E} \left[{\frac {\partial }{\partial \theta _{i}}}\log f\left(x;{{\theta }}\right){\frac {\partial }{\partial \theta _{j}}}\log f\left(x;{ {\theta }}\right)\right]=-\operatorname {E} \left[{\frac {\partial ^{2}}{\partial \theta _{i}\partial \theta _{j}}}\log f\left(x;{\boldsymbol {\theta }}\right)\right].}$$

We note that,  the Cram\'{e}r-Rao bound on the estimation error relies on some regularity conditions on the probability density function, $ f(x; \boldsymbol{\theta})$, and the estimator $ g(X)$. For a scalar parameter $\theta$, the Cram\'{e}r-Rao bound depends on two weak regularity conditions on the probability density function, $ f(x; \theta)$, and the estimator $ g(X)$. The first condition is that  the Fisher information is always defined; namely, for all $x$ such that $ f(x;\theta )>0$,
$$ {\frac  {\partial }{\partial \theta }}\log f(x;\theta )$$
exists, and is finite.
The second condition is that the operations of integration with respect to $x$ and differentiation with respect to $ \theta$  can be interchanged in the expectation of $g$, namely,
$$ {\frac  {\partial }{\partial \theta }}\left[\int g(x)f(x;\theta )\,dx\right]=\int g(x)\left[{\frac  {\partial }{\partial \theta }}f(x;\theta )\right]\,dx $$
whenever the right-hand side is finite.

For the multivariate parameter vector $ {\theta}$, the Cram\'{e}r-Rao bound then states that the covariance matrix of $   {g}(X) $ satisfies
$$ {\displaystyle \operatorname {cov} _{ {\theta }}\left({ {g}}(X)\right)\geq {\frac {\partial { {\phi }}\left({ {\theta }}\right)}{\partial { {\theta }}}}[I\left({ {\theta }}\right)]^{-1}\left({\frac {\partial { {\phi }}\left({ {\theta }}\right)}{\partial {{\theta }}}}\right)^{T}},$$
where $\frac {\partial { {\phi }}\left({{\theta }}\right)}{\partial {{\theta }}}$ is the Jacobian matrix with the element in the $i$-th row and $j$-th column as  $\partial \phi _{i}({ {\theta }})/\partial \theta _{j}$.

We define the Fisher information matrix of ${\theta}$ from the $i$-th sample of the $k$-th data source as $F^{k} (c_{i}^{k},\theta)$ (sometimes when the context is clear, we abbreviate it as $F^{k} (c_{i}^{k})$ )  Thus,  we have the Fisher information matrix of ${\theta}$  for data from the $k$-th data source as
$$\sum_{i=1}^{m_k}F^{k} (c_{i}^{k},\theta),~1\leq k \leq K,$$
since the Fisher information matrices are additive for independent observations of the same parameters.  Since we assume that the samples from different data sources are also independently generated, the Fisher information matrix of $ {\theta}$ considering all the data sources is given by 
$$F(\theta)=\sum_{k=1}^{K}\sum_{i=1}^{m_k}F^{k} (c_{i}^{k}, {\theta}),$$
which we abbreviate as $F$ in the following derivations.

Without loss of generality, we consider estimating a scalar function $T(\theta)$ of $\theta$. To estimate $T(\theta)$, we can lower bound the variance of the estimate of $T(\theta)$ by
$$ \frac{\partial \phi(\theta)}{\partial \theta} F^{-1}  (\frac{\partial \phi(\theta)}{\partial \theta})^T,$$
where $\phi(\theta)=E(g(X))$ is a scalar. If the estimate of $T(\theta)$ is unbiased, namely $\phi(\theta)=T(\theta)$, 
we can lower bound the mean-squared error of the estimate of $T(\theta)$ by
$$ \frac{\partial T(\theta)}{\partial \theta} F^{-1}  (\frac{\partial T(\theta)}{\partial \theta})^T.$$

The goal of clean sensing is to minimize the error of parameter estimation from heterogeneous data sources. Suppose we require the estimation of a function $T(\theta)$ of parameter $\theta$ to be unbiased, one way to minimize the mean-squared error of the estimation is to design sensing strategies
which minimize the Cram\'{e}r-Rao  lower bound on the estimate.  Mathematically, we are trying to solve the following optimization problem:
\begin{align}
&\min_{m_k, c_i^k} ~~~~~~\frac{\partial T(\theta)}{\partial \theta} F^{-1}  (\frac{\partial T(\theta)}{\partial \theta})^T\\
&\text{subject to}~~~~~~\sum_{k=1}^{K}\sum_{i=1}^{m_k} c_i^k \leq C,\\
&~~~~~~~~~~~~~~~~~~~F(\theta)=\sum_{k=1}^{K}\sum_{i=1}^{m_k}F^{k} (c_{i}^{k}, \theta),\\
&~~~~~~~~~~~~~~~~~~~m_k \in \mathbb{Z}_{\geq 0},\\
&~~~~~~~~~~~~~~~~~~~c_i^k\geq 0, 1\leq k \leq K, 1\leq i \leq m_k.
\label{eq:optimizationknownparameter1}
\end{align}
where $C$ is the total budget, and $\mathbb{Z}_{\geq 0}$ is the set of nonnegative integers.

We notice that this optimization depends on knowledge of the parameter vector $\theta$. However, before sampling begins, we have limited knowledge of the parameters.  Depending on the goals of estimation, we can change the objective function of  the optimization problem (\ref{eq:optimizationknownparameter1}) using the limited knowledge of the parameters.  For example, if we know in advance a prior distribution $h(\theta)$ of $\theta$, we would like to minimize the expectation of the Cram\'{e}r-Rao  lower bound for an unbiased estimator over the prior distribution. Mathematically, we formulate the corresponding optimization problem as
\begin{align}
&\min_{m_k, c_i^k} ~~~~~~ \int h(\theta) \frac{\partial T(\theta)}{\partial \theta} F^{-1}(\theta)  (\frac{\partial T(\theta)}{\partial \theta})^T \, d\theta\\
&\text{subject to}~~~~~~\sum_{k=1}^{K}\sum_{i=1}^{m_k} c_i^k \leq C,\\
&~~~~~~~~~~~~~~~~~~~F(\theta)=\sum_{k=1}^{K}\sum_{i=1}^{m_k}F^{k} (c_{i}^{k}, \theta),\\
&~~~~~~~~~~~~~~~~~~~m_k \in \mathbb{Z}_{\geq 0},\\
&~~~~~~~~~~~~~~~~~~~c_i^k\geq 0, 1\leq k \leq K, 1\leq i \leq m_k.
\label{eq:optimizationexpected}
\end{align}
where $C$ is the total budget.  If we instead know in advance the parameter to be estimated belongs to a set $\Omega$, we can also try to minimize the worst-case Cram\'{e}r-Rao  lower bound for an unbiased estimator over $\Omega$. We can thus write the minimax optimization problem as
\begin{align}
&\min_{m_k, c_i^k}\max_{\theta \in \Omega}~~~~~~  \frac{\partial T(\theta)}{\partial \theta} F^{-1}(\theta)  (\frac{\partial T(\theta)}{\partial \theta})^T\\
&\text{subject to}~~~~~~\sum_{k=1}^{K}\sum_{i=1}^{m_k} c_i^k \leq C,\\
&~~~~~~~~~~~~~~~~~~~F(\theta)=\sum_{k=1}^{K}\sum_{i=1}^{m_k}F^{k} (c_{i}^{k}, \theta),\\
&~~~~~~~~~~~~~~~~~~~m_k \in \mathbb{Z}_{\geq 0},\\
&~~~~~~~~~~~~~~~~~~~c_i^k\geq 0, 1\leq k \leq K, 1\leq i \leq m_k.
\label{eq:optimizationworstcase}
\end{align}

When the Cram\'{e}r-Rao  lower bounds $$\frac{\partial T(\theta)}{\partial \theta} F^{-1}(\theta)  (\frac{\partial T(\theta)}{\partial \theta})^T=f(m_1, m_2,...,m_K, c_1^1, c_2^1, ... , c_{m_1}^{1},..., c_{m_K}^K)$$ are the same for all the possible values of $\theta$'s under every possible choice of $c_i^k$'s, then the optimization problems  (\ref{eq:optimizationknownparameter1}), (\ref{eq:optimizationworstcase}), and (\ref{eq:optimizationexpected}) all reduce to
\begin{align}
&\min_{m_k,c_i^k} f(m_1, m_2,...,m_K, c_1^1, c_2^1, ... , c_{m_1}^{1}, c_{1}^{2}, ..., c_{m_K}^K)\\
&\text{subject to}~~~~~~\sum_{k=1}^{K}\sum_{i=1}^{m_k} c_i^k \leq C,\\
&~~~~~~~~~~~~~~~~~~~F(\theta)=\sum_{k=1}^{K}\sum_{i=1}^{m_k}F^{k} (c_{i}^{k}, \theta),\\
&~~~~~~~~~~~~~~~~~~~m_k \in \mathbb{Z}_{\geq 0},\\
&~~~~~~~~~~~~~~~~~~~c_i^k\geq 0, 1\leq k \leq K, 1\leq i \leq m_k.
\label{eq:optimizationuniformtheta}
\end{align}

\section{Optimal Sensing Strategy for Independent Heterogenous Data Sources with Diagonal Fisher Information Matrices}
\label{sec:CSsolver}

In this section, we will investigate the optimal sensing strategy for independent heterogenous data sources with diagonal Fisher information matrices.  In this section, we assume that, under every possible action $\mathcal{A}$,
 \begin{align*}
 &f( X_{1}^{1},  X_{2}^{1}, ..., , X_{{m_{1}}}^{{1}},  X_{1}^{2},  X_{2}^{2}, ..., , X_{{m_{2}}}^{{2}}, ..., X_{1}^{K}, ..., , X_{{m_{K}}}^{{K}}, \theta, \mathcal{A} )\\
 &=f_{1}^{{1}}( X_{1}^{1}, \theta_1, \mathcal{A})f_{2}^{{1}} (X^{1}_{2}, \theta_{1}, \mathcal{A}) \cdots f^{1}_{m_{1}} (X^{1}_{m_{1}}, \theta_{1}, \mathcal{A}) \\
 &~\times f_{1}^{{2}}(   X_{1}^{2}, \theta_{2}, \mathcal{A})\cdots \times f_{m_2}^{{2}}(   X_{m_2}^{2}, \theta_{2}, \mathcal{A}) \cdots \times f_{m_{K}}^{K}(X_{{m_{K}}}^{{K}}, \theta_{K}, \mathcal{A}).
\end{align*}

One can show that under this assumption, the Fisher information matrices
$$F(\theta)=\sum_{k=1}^{K}\sum_{i=1}^{m_k}F^{k} (c_{i}^{k}, \theta)$$
is  a diagonal matrix, where $F^{k} (c_{i}^{k}, \theta)$ is a $d \times d$ Fisher information matrix based on observation $X_{i}^{k}$.  Moreover,  for every $k$ and $i$,  $F^{k} (c_{i}^{k}, \theta)$ is a  diagonal matrix, for which only the $k$-th element of the diagonal, denoted by $F^{k} (c_{i}^{k}, \theta_{k})$, can possibly be nonzero.

\begin{theorem}
Let us consider estimating a function $T(\theta)$ of an unknown parameter vector $\theta$ of dimension $d$, using data from $K$ independent heterogenous data sources, where $K$ is a positive integer.  From the $k$-th data source, we obtain $m_{{k}}$ samples, where $1\leq k \leq K$.   We denote the $m_{k}$ samples from the $k$-th data source as $X_{1}^{k}$,  $X_{2}^{k}$, ..., and $X_{{m_{k}}}^{{k}}$.  We assume that cost $c^{k}_{i}$ was spent on acquiring the $i$-th sample from the $k$-th data source, where $1\leq k \leq K$, and $1\leq i \leq m_{k}$. We further assume that  $X_{1}^{k}$,  $X_{2}^{k}$, ..., and $X_{{m_{k}}}^{{k}}$ are mutually independent.

We let $F^{k} (c_{i}^{k}, \theta)$ be a $d \times d$ Fisher information matrix based on observation $X_{i}^{k}$, as a function of $\theta$ and cost $c_{i}^{k}$.   We assume that $X_{i}^{k}$ only reveals information about $\theta_{k}$; namely, we assume that,  for every $k$ and $i$,  $F^{k} (c_{i}^{k}, \theta)$ is a  diagonal matrix, for which only the $k$-th element of the diagonal, denoted by  $F^{k} (c_{i}^{k}, \theta_{k})$ as a function of $c_{i}^{k}$ and $\theta_{k}$, can possibly be nonzero.

We assume,  under a cost $c$,  the function $F^{k} (c, \theta_{k})$ satisfies the following conditions:
\begin{enumerate}
\item $F^{k} (c, \theta_{k})$ is a non-decreasing function in $c$ for $c\geq 0$;
\item $\frac{c}{  F^{k} (c, \theta_{k})}$ is well defined for $c\geq 0$;
\item $\frac{c}{  F^{k} (c, \theta_{k})}$ achieves its minimum value at a finite $c=c_k^{**} \geq 0$, and the corresponding minimum value $\frac{c_{k}^{**}}{  F^{k} (c_{k}^{**}, \theta_{k})}$ is denoted by $p_k^*$.
\end{enumerate}

Let $g(X)$ be an unbiased estimator of  a function $T(\theta)$ of parameter vector $\theta$. Let $C$ be the total allowable budget for acquiring samples from the $K$ data sources. When  $C\rightarrow \infty$,  the smallest possible achievable Cram\'{e}r-Rao lower bound $\frac{\partial T(\theta)}{\partial \theta} F^{-1}(\theta)  (\frac{\partial T(\theta)}{\partial \theta})^T$ on the mean-squared error of $g(X)$ is given by

$$\frac{1}{C} \left(  \sum_{k=1}^{{K}}\frac{\partial T(\theta)}{\partial \theta_k}\sqrt{p_k^*} \right)^2.$$

Moreover, when  $C\rightarrow \infty$, for the optimal sensing strategy that achieves this smallest possible Cram\'{e}r-Rao lower bound on the the mean-squared error, we have the optimal cost allocated for the $k$-th data source, denoted by $C_{k}^{*}$, satisfies:
$$C_{k}^*= \frac{C   \frac{\partial T(\theta)}{\partial \theta_k} \sqrt{p_k^*}}{\sum_{k=1}^{{K}}\frac{\partial T(\theta)}{\partial \theta_k}\sqrt{p_k^*}}. $$

The optimal cost $(c_{i}^{k})^{*}$ associated with obtaining the $i$-th sample of the $k$-th source is given by

$$ (c_{i}^{k})^{*}=c_{k}^{**}.   $$

The number of samples obtained from the $k$-th data source satisfies
$$m_{k}^{*}=\left (\frac{C  \frac{\partial T(\theta)}{\partial \theta_k} \sqrt{ p_k^*}}{\sum_{k=1}^{{K}}\frac{\partial T(\theta)}{\partial \theta_k}\sqrt{p_k^*}}\right) /{c_k^{**}}.$$

\label{theorem:main}
\end{theorem}

\begin{proof}
We assume that budget $C_{k}$ is allocated to obtain $m_{k}$ samples from the $k$-th data source, namely
$$C_{k}=\sum_{i=1}^{m_k} c_{i}^{k}. $$

Then we can conclude that
\begin{align}
&~~ \sum_{i=1}^{m_k}F^{k} (c_{i}^{k}, \theta_{k})     \\
&=\sum_{i=1}^{m_k} {c_i^k}/\left(\frac{c_i^k}{ F^{k} (c_{i}^{k}, \theta_{k}) } \right) \\
& \leq  \sum_{i=1}^{m_k} \frac{c_i^k}{p_{k}^{*}} \label{eq:middle}\\
&= \frac{\sum_{i=1}^{m_k}c_i^{k}}{p_{k}^{*}}\\
&=\frac{C_{k}}{p_{k}^{*}},
\end{align}
where in (\ref{eq:middle}) we use the fact that $\frac{c_{k}}{  F^{k} (c^{k}, \theta_{k})}$ achieves its minimum value at a finite $c_k^{**} \geq 0$, and the corresponding minimum value $\frac{c_{k}^{**}}{  F^{k} (c_{k}^{**}, \theta_{k})}$ is denoted by $p_k^*$.

Moreover, we claim that there exists a strategy of allocating budget $C_{k}$ to samples in such a way that
\begin{align}
&~~ \sum_{i=1}^{m_k}F^{k} (c_{i}^{k}, \theta_{k})     \\
&\geq \frac{C_{k}}{p_{k}^{*}}-  F^{k} (c_{k}^{**}, \theta_{k}).
\end{align}

In fact, one can take $$m_{k}=\floor{  \frac{C_{k}}{c_{k}^{**}}  }$$ samples,  and we spend $c_{k}^{**}$ to obtain each of the samples except for the last sample, on which we spend $ \left ( \frac{C_{k}}{  c_{k}^{**}  } -\floor{  \frac{C_{k}}{  c_{k^{**}}  }   } \right) \times c_{k}^{{**}}$. Then we have
\begin{align}
&~~ \sum_{i=1}^{m_k}F^{k} (c_{i}^{k}, \theta_{k}) \\
&\geq  \floor{  \frac{C_{k}}{  c_{k}^{**}  }   } F^{k} (c_{k}^{**}, \theta_{k})\\
&\geq  \left(\frac{C_{k}}{ c_{k}^{**}  } -1 \right) F^{k} (c_{k}^{**}, \theta_{k})\\
& =\frac{C_{k}  F^{k} (c_{k}^{**}, \theta_{k})   }{c_{k}^{**}}-F^{k} (c_{k}^{**}, \theta_{k})\\
&=\frac{C_{k}  }{p_{k}^{*}}-F^{k} (c_{k}^{**}, \theta_{k}).
\end{align}

In summary, we have

$$ \max\left\{0,  \frac{C_{k}  }{p_{k}^{*}}-F^{k} (c_{k}^{**}, \theta_{k})\right\} \leq   \sum_{i=1}^{m_k}F^{k} (c_{i}^{k}, \theta_{k})  \leq   \frac{C_{k}}{p_{k}^{*}}.$$

Then the smallest possible achievable Cram\'{e}r-Rao lower bound $\frac{\partial T(\theta)}{\partial \theta} F^{-1}(\theta)  (\frac{\partial T(\theta)}{\partial \theta})^T$, as defined by the optimal objective function of the following optimization problem,
\begin{align}
&\min_{m_k,c_i^k} \frac{\partial T(\theta)}{\partial \theta} F^{-1}(\theta)  (\frac{\partial T(\theta)}{\partial \theta})^T\\
&\text{subject to}~~~~~~\sum_{k=1}^{K}\sum_{i=1}^{m_k} c_i^k \leq C,\\
&~~~~~~~~~~~~~~~~~~~F(\theta)=\sum_{k=1}^{K}\sum_{i=1}^{m_k}F^{k} (c_{i}^{k}, \theta),\\
&~~~~~~~~~~~~~~~~~~~m_k \in \mathbb{Z}_{\geq 0},\\
&~~~~~~~~~~~~~~~~~~~c_i^k\geq 0, 1\leq k \leq K, 1\leq i \leq m_k.
\label{eq:optimizationworstcase111111}
\end{align}
can be lower bounded by the optimal value of the following optimization problem:
\begin{align}
&\min_{C_{k}} \frac{\partial T(\theta)}{\partial \theta} F^{-1}(\theta)  (\frac{\partial T(\theta)}{\partial \theta})^T\\
&\text{subject to}~~~~~~\sum_{k=1}^{K}C_k \leq C,\\
&~~~~~~~~~~~~~~~~~~~F(\theta)=\begin{bmatrix}
  {\frac{C_{1}}{p_{1}^{*}}}& 0 &0&\cdots &0&\\
0& \frac{C_{2}}{p_{2}^{*}}  &0&\cdots &0\\
0& 0& \frac{C_{3}}{p_{3}^{*}} &\cdots&0\\
 \vdots &\vdots &\vdots &\ddots &\vdots \\
0& 0& 0&\cdots &\frac{C_{K}}{p_{K}^{*}}
 \end{bmatrix}, \\
&~~~~~~~~~~~~~~~~~~~C_{k}\geq 0, 1\leq k \leq K.
\label{eq:optimizationworstcase and simplified1}
\end{align}

We can solve (\ref{eq:optimizationworstcase and simplified1})  through its Lagrange dual problem and the Karush-Kuhn-Tucker conditions  (please refer to Appendix \ref{appendixKKT1}).  For the optimal solution, we have obtained that
$$C_{k}^*= \frac{C   \frac{\partial T(\theta)}{\partial \theta_k} \sqrt{p_k^*}}{\sum_{k=1}^{{K}}\frac{\partial T(\theta)}{\partial \theta_k}\sqrt{p_k^*}}, $$
and, moreover, under the optimal $C_{k}^{*}$,  the optimal value of (\ref{eq:optimizationworstcase and simplified1})  is given by
$$\frac{1}{C} \left(  \sum_{k=1}^{{K}}\frac{\partial T(\theta)}{\partial \theta_k}\sqrt{p_k^*} \right)^2.$$

Now let us consider upper bounding the smallest possible achievable Cram\'{e}r-Rao lower bound $\frac{\partial T(\theta)}{\partial \theta} F^{-1}(\theta)  (\frac{\partial T(\theta)}{\partial \theta})^T$, as defined by the following optimization problem,
\begin{align}
&\min_{m_k,c_i^k} \frac{\partial T(\theta)}{\partial \theta} F^{-1}(\theta)  (\frac{\partial T(\theta)}{\partial \theta})^T\\
&\text{subject to}~~~~~~\sum_{k=1}^{K}\sum_{i=1}^{m_k} c_i^k \leq C,\\
&~~~~~~~~~~~~~~~~~~~F(\theta)=\sum_{k=1}^{K}\sum_{i=1}^{m_k}F^{k} (c_{i}^{k}, \theta),\\
&~~~~~~~~~~~~~~~~~~~m_k \in \mathbb{Z}_{\geq 0},\\
&~~~~~~~~~~~~~~~~~~~c_i^k\geq 0, 1\leq k \leq K, 1\leq i \leq m_k.
\label{eq:optimizationworstcase2}
\end{align}

We note that, because
$$   \sum_{i=1}^{m_k}F^{k} (c_{i}^{k}, \theta_{k}) \geq \max\left\{0,  \frac{C_{k}  }{p_{k}^{*}}-F^{k} (c_{k}^{**}, \theta_{k})\right\},$$
the optimal objective value of (\ref{eq:optimizationworstcase2}) can be upper bounded by the optimal objective value of the following optimization problem (\ref{eq:optimizationworstcase with residue1st}):
\begin{align}
&\min_{C_{k}, t_{k}} \frac{\partial T(\theta)}{\partial \theta} F^{-1}(\theta)  (\frac{\partial T(\theta)}{\partial \theta})^T\nonumber\\
&\text{subject to}~~~~~~\sum_{k=1}^{K}C_k \leq C,\nonumber\\
&~~~~~~~~~~~~~~~~~~~F(\theta)=\begin{bmatrix}
  t_{1}& 0 &0&\cdots &0&\\
0& t_{2}  &0&\cdots &0\\
0& 0& t_{3} &\cdots&0\\
 \vdots &\vdots &\vdots &\ddots &\vdots \\
0& 0& 0&\cdots &t_{K}
 \end{bmatrix},\nonumber\\
&~~~~~~~~~~~~~~~~~~~t_{k} \leq \max \left( \left (\frac{C_{k}}{p_{k}^{*}} -F^{k} (c_{k}^{**}, \theta_{k}) \right), 0 \right), 1\leq k \leq K,\nonumber\\
&~~~~~~~~~~~~~~~~~~~t_{k} \geq 0, 1\leq k \leq K,\nonumber\\
&~~~~~~~~~~~~~~~~~~~C_{k}\geq 0, 1\leq k \leq K.
\label{eq:optimizationworstcase with residue1st}
\end{align}

We further notice that the optimal objective value of  (\ref{eq:optimizationworstcase with residue1st}) is  upper bounded by the objective value of the following optimization problem (\ref{eq:optimizationworstcase with residue1st}):

\begin{align}
&\min_{C_{k}, t_{k}} \frac{\partial T(\theta)}{\partial \theta} F^{-1}(\theta)  (\frac{\partial T(\theta)}{\partial \theta})^T\nonumber\\
&\text{subject to}~~~~~~\sum_{k=1}^{K}C_k \leq C,\nonumber\\
&~~~~~~~~~~~~~~~~~~~F(\theta)=\begin{bmatrix}
  t_{1}& 0 &0&\cdots &0&\\
0& t_{2}  &0&\cdots &0\\
0& 0& t_{3} &\cdots&0\\
 \vdots &\vdots &\vdots &\ddots &\vdots \\
0& 0& 0&\cdots &t_{K}
 \end{bmatrix},\nonumber\\
&~~~~~~~~~~~~~~~~~~~t_{k} = \frac{C_{k}}{p_{k}^{*}} -F^{k} (c_{k}^{**}, \theta_{k}) , 1\leq k \leq K,\nonumber\\
&~~~~~~~~~~~~~~~~~~~\frac{C_{k}}{p_{k}^{*}} -F^{k} (c_{k}^{**}, \theta_{k})  \geq 0,  1\leq k \leq K,    \nonumber\\
&~~~~~~~~~~~~~~~~~~~C_{k}\geq 0, 1\leq k \leq K.
\label{eq:optimizationworstcase with residue}
\end{align}

One can solve (\ref{eq:optimizationworstcase with residue})  (please refer to Appendix   \ref{appendixKKT2}), and obtain that
$$C_{k}= \frac{\left(C-\sum_{k=1}^{K} c_{k}^{**}\right) \frac{\partial T(\theta)}{\partial \theta_k} \sqrt{ p_k^*}}{\sum_{k=1}^{{K}}\frac{\partial T(\theta)}{\partial \theta_k}\sqrt{p_k^*}}+c_{k}^{**}. $$

Moreover, the optimal value of (\ref{eq:optimizationworstcase with residue})  is given by
$$\frac{1}{ C-\sum_{k=1}^{K}c_{k}^{**}  } \left(  \sum_{k=1}^{{K}}\frac{\partial T(\theta)}{\partial \theta_k}\sqrt{p_k^*} \right)^2.$$

Under this strategy, the number of samples $m_k$ satisfies

$$m_{k}=\frac{C_k}{c_k^{**}}+o(C)=\frac{\frac{C   \frac{\partial T(\theta)}{\partial \theta_k}  \sqrt{p_k^*}}{\sum_{k=1}^{{K}}\frac{\partial T(\theta)}{\partial \theta_k}\sqrt{p_k^*}}}{c_k^{**}}+o(C).$$

\end{proof}

As can be seen from the conclusions of Theorem \ref{theorem:main},  the cost allocated for sensing from the $k$-th data source is not the same for each data source: it is proportional how important $\theta_{k}$ is to the estimation goal (namely $\frac{\partial T(\theta)}{\partial \theta_k}$), and also is proportional to the ``inverse'' data quality of the $k$-th source (namely the square root of $p_{k}^{*}$, which is the optimal ``cost-over-Fisher-information-ratio'' for the $k$-th data source). We note that the higher $p_{k}^{*}$ is, the worse the data quality is, and the harder to obtain a certain amount of Fisher information from the $k$-th data source.  Asymptotically, out results show that the cost spent on obtaining each data sample from the $k$-th data source should be $c_{k}^{**}$, namely the cost at which the best ``Fisher-information-over-cost-ratio'' is achieved. This is in contrast to traditional practices of using the same cost for each sample from each data source.  We also observe that, asymptotically,  the Fisher information provided by the $k$-th data source is given by 
$$ \frac{C_{k}}{p_{k}^{**}}=(1+o(1)) \frac{C   \frac{\partial T(\theta)}{\partial \theta_k}  \sqrt{p_k^*}}{\sum_{k=1}^{{K}}\frac{\partial T(\theta)}{\partial \theta_k}\sqrt{p_k^*}} \times \frac{1}{p_{k}^{*}}=\frac{C   \frac{\partial T(\theta)}{\partial \theta_k}  /\sqrt{p_k^*}}{\sum_{k=1}^{{K}}\frac{\partial T(\theta)}{\partial \theta_k}\sqrt{p_k^*} } ,$$
which implies the Fisher information eventually provided by the $k$-th data source should be proportional to the square root of ``Fisher-information-over-cost-ratio''. This means that the optimal sensing strategy should let the sources with better data qualities eventually provide more Fisher information (assuming the $K$ parameters are of the same level importance to the estimation objective).  
The number of samples $m_{k}$ from the $k$-th data source should be given by 
$$m_{k}=\frac{\frac{C   \frac{\partial T(\theta)}{\partial \theta_k}  \sqrt{p_k^*}}{\sum_{k=1}^{{K}}\frac{\partial T(\theta)}{\partial \theta_k}\sqrt{p_k^*}}}{c_k^{**}}= \frac{ C   \frac{\partial T(\theta)}{\partial \theta_k}  \sqrt{\frac{1}{ F^{k} (c_{k}^{**}, \theta_{k})  c_{k}^{**}    }}}{\sum_{k=1}^{{K}}\frac{\partial T(\theta)}{\partial \theta_k}\sqrt{p_k^*}} .$$
This means the number of samples from a data source is inversely proportional to the square root of  the ``Fisher-information-cost-product'' at the best individual sample cost for that source.

In summary, (assuming that  the $K$ parameters are equally important to the estimation objective),   the higher data quality the $k$-th data source provides, the less total budget should be allocated for sampling from the $k$-th source, and the more eventual Fisher information the $k$-th data source will provide.  To further understand the implications of our results,  let us consider the special case $F^{1} (c_{1}^{**}, \theta_{1})=F^{2} (c_{2}^{**}, \theta_{2})=\cdots=F^{K} (c_{K}^{**}, \theta_{K})$, and $\frac{\partial T(\theta)}{\partial \theta_1}=\frac{\partial T(\theta)}{\partial \theta_2}=\cdots=\frac{\partial T(\theta)}{\partial \theta_K}$.  Then the total sensing budget allocated for the $k$-th data source should be proportional to $\sqrt{c_{k}^{**}}$,  the number of samples from the $k$-th data source should be  proportional to $\frac{1}{\sqrt{ c_{k}^{**} }}$, and the cost spent on each sample from the $k$-th data source should be proportional to (equal to) $c_{k}^{**} $. This implies, that under this special case, the worse data quality the $k$-th data source provides (namely, for the same Fisher information from an individual sample, a higher cost needs to be spent on a sample), the more sensing budget should be allocated for the $k$-th data source, a higher cost should be spent on obtaining an individual sample from the $k$-th data source, but a smaller number samples should be taken from the $k$-th data source.

\section{Clean Sensing for Optimal Parameter Estimation of Gaussian Random Variables}
\label{sec:Gaussian}

In the last section, we have derived the optimal sensing strategy for independent random variables from heterogenous data sources.  As discussed above,  we can extend the framework of clean sensing to estimate parameters from dependent random variables.  In this section, we will derive the optimal sensing strategies to estimate the parameters related to the mean values of multivariate Gaussian random variables, which are not necessarily independent.   In the last section, we have mostly considered the case where data from the $k$-th data source are only concerned with the $k$-th parameter $\theta_{k}$. In this section,  we have extended the clean sensing framework to include cases where data from the $k$-th data source may provide information for more than just one parameter, and sometimes even for all the parameters.  Moreover, for the case of multivariate Gaussian random variables, we have closed-form expressions for the related Fisher information matrix, and the examples in this section illustrate how the clean sensing framework can be applied to signal processing examples of parameter estimation for Gaussian random variables.


We consider $K$ Gaussian random variables ${\displaystyle X={\begin{bmatrix}X_{1},\dots ,X_{K}\end{bmatrix}}^{\mathrm {T} }}$.  We assume that the mean values of these random variables are $  {\displaystyle \,\mu (\theta )={\begin{bmatrix}\mu _{1}(\theta ),\dots ,\mu _{K}(\theta )\end{bmatrix}}^{\mathrm {T} }}$, where $\theta$ is a $K$-dimensional vector  (we can also consider vectors of other dimensions, but we choose $K$ to simplify the analysis). We let $ {\displaystyle \,\Sigma (\theta )}$ be its covariance matrix. Then, for ${\displaystyle 1\leq m,n\leq K}$,  the  element in the $m$-th row and the $n$-th column of the Fisher information matrix (with respect to parameter $\theta$) is given by \cite{Kay}:

$$F_{m,n}
=
\frac{\partial \mu^\mathrm{T}}{\partial \theta_m}
\Sigma^{-1}
\frac{\partial \mu}{\partial \theta_n}
+
\frac{1}{2}
\operatorname{tr}
\left(
 \Sigma^{-1}
 \frac{\partial \Sigma}{\partial \theta_m}
 \Sigma^{-1}
 \frac{\partial \Sigma}{\partial \theta_n}
\right),$$
where $(\cdot )^{\mathrm {T} }$ denotes the transpose of a vector, $\operatorname{tr} (\cdot)$ denotes the trace of a square matrix, and

$${\displaystyle {\frac {\partial \mu }{\partial \theta _{m}}}={\begin{bmatrix}{\frac {\partial \mu _{1}}{\partial \theta _{m}}}&{\frac {\partial \mu _{2}}{\partial \theta _{m}}}&\cdots &{\frac {\partial \mu _{K}}{\partial \theta _{m}}}\end{bmatrix}}^{\mathrm {T} };}$$

$$
 {\displaystyle {\frac {\partial \Sigma }{\partial \theta _{m}}}={\begin{bmatrix}{\frac {\partial \Sigma _{1,1}}{\partial \theta _{m}}}&{\frac {\partial \Sigma _{1,2}}{\partial \theta _{m}}}&\cdots &{\frac {\partial \Sigma _{1,K}}{\partial \theta _{m}}}\\[5pt]{\frac {\partial \Sigma _{2,1}}{\partial \theta _{m}}}&{\frac {\partial \Sigma _{2,2}}{\partial \theta _{m}}}&\cdots &{\frac {\partial \Sigma _{2,K}}{\partial \theta _{m}}}\\\vdots &\vdots &\ddots &\vdots \\{\frac {\partial \Sigma _{K,1}}{\partial \theta _{m}}}&{\frac {\partial \Sigma _{K,2}}{\partial \theta _{m}}}&\cdots &{\frac {\partial \Sigma _{K,K}}{\partial \theta _{m}}}\end{bmatrix}}.}$$

Note that a special case is the one where $ {\displaystyle \Sigma (\theta )=\Sigma }$ is a constant matrix. When $ \Sigma (\theta )$ is a constant matrix, we have
$$
F_{m,n}=
\frac{\partial \mu^\mathrm{T}}{\partial \theta_m}
\Sigma^{-1}
\frac{\partial \mu}{\partial \theta_n}.$$

Namely, the Fisher information matrix $F$ can be written as
$$F =   \frac{\partial \mu}{\partial \theta}^{T}  \Sigma^{-1}   \frac{\partial \mu}{\partial \theta},  $$
where the $K \times K$ matrix $\frac{\partial \mu}{\partial \theta}$ is defined as
\begin{align}
 {\displaystyle {\frac {\partial \mu }{\partial \theta}}={\begin{bmatrix}{\frac {\partial \mu _{1}}{\partial \theta _{1}}}&{\frac {\partial \mu _{1}}{\partial \theta _{2}}}&\cdots &{\frac {\partial \mu _{1}}{\partial \theta _{K}}}\\[5pt]{\frac {\partial \mu _{2}}{\partial \theta _{1}}}&{\frac {\partial \mu _{2}}{\partial \theta _{2}}}&\cdots &{\frac {\partial \mu _{2}}{\partial \theta _{K}}}\\\vdots &\vdots &\ddots &\vdots \\{\frac {\partial \mu _{K}}{\partial \theta _{1}}}&{\frac {\partial \mu _{K}}{\partial \theta _{2}}}&\cdots &{\frac {\partial \mu _{K}}{\partial \theta _{K}}}\end{bmatrix}}.}
 \label{eq:definitionof_partial_u_theta}
 \end{align}

Let $g(X)$ be an unbiased estimator of  a function $T(\theta)$ of parameter vector $\theta$.  Assuming $F$ is invertible, we can lower bound the mean-squared error of the estimate $g(X)$ of $T(\theta)$ by
\begin{align}
 &~~~~\frac{\partial T(\theta)}{\partial \theta}^T F^{-1}  (\frac{\partial T(\theta)}{\partial \theta})\\
 &~=\frac{\partial T(\theta)}{\partial \theta}^T    \left ( \frac{\partial \mu}{\partial \theta}^{T}  \Sigma^{-1}   \frac{\partial \mu}{\partial \theta}   \right)^{-1}          \frac{\partial T}{\partial \theta}\\
 &~=\frac{\partial T(\theta)}{\partial \theta}^T    \left ({\frac{\partial \mu}{\partial \theta}}\right )^{-1}  \Sigma   \left( {\frac{\partial \mu}{\partial \theta}}^{T}\right)^{-1}             \frac{\partial T}{\partial \theta}.
 \end{align}

We denote $v^{T}=\frac{\partial T(\theta)}{\partial \theta}^T    \left ({\frac{\partial \mu}{\partial \theta}}\right )^{-1}$, then the lower bound of the mean-squared error of the estimate $g(X)$ of $T(\theta)$ is given by
\begin{align}
 &~~~~\frac{\partial T(\theta)}{\partial \theta}^T F^{-1}  (\frac{\partial T(\theta)}{\partial \theta})\\
 &~=v^{T}  \Sigma   v .
 \end{align}

\subsection{Optimal Strategy for One-Time Sampling of Gaussian Random Variables}
Suppose that we consider the case where we take one sample from each of the $K$ data sources.
We assume that if we spend cost $C_{i}$ on sampling data source $i$, where $1\leq i \leq K$, the corresponding covariance matrix $\Sigma$ is given by

$$\Sigma= \sum_{i=1}^{K} \frac{1}{C_{i}}v_{i} v_{i}^{T},$$
where $v_{i}$'s are known vectors.   Then the lower bound of the mean-squared error of the estimate of $T(\theta)$ is given by
\begin{align}
 &~~~~\frac{\partial T(\theta)}{\partial \theta}^T F^{-1}  (\frac{\partial T(\theta)}{\partial \theta})\\
 &~=v^{T}  \sum_{i=1}^{K} \frac{1}{C_{i}}v_{i} v_{i}^{T}   v\\
 &~=\sum_{i=1}^{K}  \frac{\|v_{i}^{T} v \|^{2}}{C_{i}}.
 \end{align}

Then to minimize the mean-squared error of the estimation under the total cost constraint, we need to solve the following optimization problem:
\begin{align}
&\overset{\text{minimize}}{C_{k}} ~~~~~~~~~~\sum_{i=1}^{K}  \frac{\|v_{i}^{T} v \|^{2}}{C_{i}} \\
&\text{subject to}~~~~~~\sum_{k=1}^{K}C_k \leq C,\\
&~~~~~~~~~~~~~~~~~~~C_{k}\geq 0, 1\leq k \leq K.
\label{eq:Gaussianonetime}
\end{align}

One can obtain the solution to (\ref{eq:Gaussianonetime}) as
$$C_{i}^{*}=  \frac{C  \|v_{i}^{T} v \|   }{\sum_{j=1}^{K}   \|v_{j}^{T} v \|    }, ~~1\leq i \leq K, $$
and the corresponding smallest lower bound is given by
$$   \sum_{i=1}^{K}  \frac{\|v_{i}^{T} v \|^{2}}{C_{i}^{*}}=\frac{1}{C} (\sum_{j=1}^{K}   \|v_{j}^{T} v \|  )^{2},$$
where $v^{T}=\frac{\partial T(\theta)}{\partial \theta}^T    \left ({\frac{\partial \mu}{\partial \theta}}\right )^{-1}.$

\subsection{Optimal Strategy for Multiple-Time Samplings of Gaussian Random Variables}
In this subsection, we assume that there are $K$ data sources, and from each data source, we take $m_{k}$ samples.  We assume that the $i$-th sample from the $k$-th data source is given by
$$ X_{i}^{k} =\mu_{i}^{k} (\theta)+w_{i}^{k},$$
where $1\leq i \leq m_{k}$, and $w_{i}^{k}$ is a zero-mean Gaussian random variable with variance $\sigma^{2}_{k, i}$.
We assume that the random variables $w_{i}^{k}$'s are independent from each other.  We further assume that we spend cost $c_{i}^{k}$  in obtaining the $i$-th sample from the $k$-th data source, and the variance is given
by
$$   \sigma^{2}_{k, i} =\frac{\sigma_{k}^{2}}{f_{k}(c_{i}^{k})},  $$
where $\sigma^{2}$ is a constant, and $f_{k}(c_{i}^{k})$ is a non-decreasing function of $c_{i}^{k}$.

From the discussions above, the Fisher information matrix $F$ with respect to $\theta$  is given by
$$F =   \frac{\partial \mu}{\partial \theta}^{T}  \Sigma^{-1}   \frac{\partial \mu}{\partial \theta},  $$
where the $K \times K$ matrix $\frac{\partial \mu}{\partial \theta}$ is defined as
$$
  \frac {\partial \mu }{\partial \theta}=\begin{bmatrix}

  \frac {\partial \mu_{1}^{1}}{\partial \theta _{1}}&{\frac{\partial \mu_{1}^{1}}{\partial \theta _{2}}}&\cdots &{\frac {\partial \mu_{1}^{1}}{\partial \theta _{K}}}\\
 \frac{\partial \mu_{2}^{1}}{\partial \theta _{1}}&{\frac {\partial \mu_{2}^{1}}{\partial \theta _{2}}}&\cdots &{\frac {\partial \mu_{2}^{1}}{\partial \theta _{K}}}\\
 \vdots &\vdots &\ddots &\vdots \\
 \frac{\partial \mu_{m_{1}}^{1}}{\partial \theta _{1}}&{\frac {\partial \mu _{m_{1}}^{1}}{\partial \theta _{2}}}&\cdots &{\frac {\partial \mu _{m_{1}}^{1}}{\partial \theta _{K}}}\\

 \frac {\partial \mu_{1}^{2}}{\partial \theta _{1}}&{\frac{\partial \mu_{1}^{2}}{\partial \theta _{2}}}&\cdots &{\frac {\partial \mu_{1}^{2}}{\partial \theta _{K}}}\\
 \frac{\partial \mu_{2}^{2}}{\partial \theta _{1}}&{\frac {\partial \mu_{2}^{2}}{\partial \theta _{2}}}&\cdots &{\frac {\partial \mu_{2}^{2}}{\partial \theta _{K}}}\\
 \vdots &\vdots &\ddots &\vdots \\
 \frac{\partial \mu_{m_{2}}^{2}}{\partial \theta _{1}}&{\frac {\partial \mu _{m_{2}}^{2}}{\partial \theta _{2}}}&\cdots &{\frac {\partial \mu _{m_{2}}^{2}}{\partial \theta _{K}}}\\
 \frac {\partial \mu_{1}^{3}}{\partial \theta _{1}}&{\frac{\partial \mu_{1}^{3}}{\partial \theta _{2}}}&\cdots &{\frac {\partial \mu_{1}^{3}}{\partial \theta _{K}}}\\
 \vdots &\vdots &\ddots &\vdots \\
 {\frac {\partial \mu _{K}}{\partial \theta _{1}}}&{\frac {\partial \mu _{K}}{\partial \theta _{2}}}&\cdots &{\frac {\partial \mu _{K}}{\partial \theta _{K}}}
 \end{bmatrix},$$
and
$\Sigma$ is an $m\times m$ diagonal matrix  defined as follows:
$$
\Sigma=\begin{bmatrix}
 \frac {\sigma_{1}^{2}}{f_{1}(c_{1}^{1})}& 0 &0&\cdots &0&0&0&0&0\\
0& \frac {\sigma_{1}^{2}}{f_{1}(c_{2}^{1})}  &0&\cdots &0&0&0&0&0\\
0& 0& \frac {\sigma_{1}^{2}}{f_{1}(c_{3}^{1})} &\cdots&0&0&0&0&0\\
 \vdots &\vdots &\vdots &\ddots &\vdots &\vdots&\vdots&\vdots&\vdots\\
0& 0& 0&\cdots &\frac {\sigma_{1}^{2}}{f_{1}(c_{m_{1}}^{1})} &\cdots &0&0&0\\
0& 0& 0&\cdots &0&\frac {\sigma_{2}^{2}}{f_{2}(c_{1}^{2})} &\cdots &0&0\\
0& 0& 0&\cdots &0&0&\frac {\sigma_{2}^{2}}{f_{2}(c_{2}^{2})} &\cdots &0\\
 \vdots &\vdots &\vdots &\vdots &\vdots &\vdots&\vdots&\ddots&\vdots\\
0& 0& 0&\cdots &0&0&\cdots&\cdots &\frac {\sigma_{K}^{2}}{f_{K}(c_{m_{K}}^{K})}
 \end{bmatrix},$$
with $m=\sum_{k=1}^{K} m_{k}$.

Then we can further express the Fisher information matrix $F$ as
\begin{align}
&F=\frac{\partial \mu}{\partial \theta}^{T}  \Sigma^{-1}   \frac{\partial \mu}{\partial \theta}\\
&~~=\sum_{k=1}^{K}\sum_{i=1}^{m_k} \left( \frac{\partial \mu_{i}^{k}}{\partial \theta}  \right) \left( \frac{\partial \mu_{i}^{k}}{\partial \theta}  \right)^{T} \frac{f_{k}(c_i^k)}{\sigma_k^2},\\
\end{align}
where
$${\displaystyle {\frac {\partial \mu_{i}^{k} }{\partial \theta}}={\begin{bmatrix}{\frac {\partial \mu _{i}^{k}}{\partial \theta_{1}}}&{\frac {\partial \mu _{i}^{k}}{\partial \theta _{2}}}&\cdots &{\frac {\partial \mu _{i}^{k}}{\partial \theta_{K}}}\end{bmatrix}}^{\mathrm {T} }}.$$

We assume that $\mu_i^k$=$\mu^k$ for every $i$ such that $1\leq i \leq m_k$. Then
\begin{align}
&F=\sum_{k=1}^{K}\sum_{i=1}^{m_k} \left( \frac{\partial \mu_{i}^{k}}{\partial \theta}  \right) \left( \frac{\partial \mu_{i}^{k}}{\partial \theta}  \right)^{T} \frac{f_{k}(c_i^k)}{\sigma_k^2},\\
&~~=\sum_{k=1}^{K}( \frac{\sum_{i=1}^{m_k}f_{k}(c_i^k)}{\sigma_k^2}   ) \left( \frac{\partial \mu^{k}}{\partial \theta}  \right) \left( \frac{\partial \mu^{k}}{\partial \theta}  \right)^{T}\\
&~~=\frac{\partial \mu}{\partial \theta}^{T}  B   \frac{\partial \mu}{\partial \theta},
\end{align}
where the $K \times K$ matrix $\frac{\partial \mu}{\partial \theta}$ is defined as in (\ref{eq:definitionof_partial_u_theta}) and
$$ B= \begin{bmatrix}
 \frac{\sum_{i=1}^{m_1}f_{1}(c_i^1)}{\sigma_1^2}& 0 &0&\cdots &0\\
0& \frac{\sum_{i=1}^{m_2}f_{2}(c_i^2)}{\sigma_2^2}  &0&\cdots &0\\
0& 0& \frac{\sum_{i=1}^{m_3}f_{3}(c_i^3)}{\sigma_3^2} &\cdots&0\\
 \vdots &\vdots &\vdots &\ddots &\vdots \\
0& 0& 0&\cdots &\frac{\sum_{i=1}^{m_K}f_{K}(c_i^K)}{\sigma_K^2} \\
 \end{bmatrix},        $$

We denote $v^{T}=\frac{\partial T(\theta)}{\partial \theta}^T    \left ({\frac{\partial \mu}{\partial \theta}}\right )^{-1}$, then   the  Cram\'{e}r-Rao  lower bound of the mean squared error of the estimate of $T(\theta)$ is given by
\begin{align}
 &~~~~\frac{\partial T(\theta)}{\partial \theta}^T F^{-1}  (\frac{\partial T(\theta)}{\partial \theta})\\
 &~=v^{T}  B^{-1}   v \\
 &~=\sum_{k=1}^{K}  \frac{\|v_{k}\|^{2} \sigma_k^2   }{ {\sum_{i=1}^{m_k}f_{k}(c_i^k)}}
 \end{align}


Then to minimize the mean-squared error of the estimation under the total cost constraint, we need to solve the following optimization problem:
\begin{align}
&\overset{\text{minimize}}{m_k, c_i^k} ~~~~~~~~~~\sum_{k=1}^{K}  \frac{\|v_{k}\|^{2} { \sigma_k^2   }}{ {\sum_{i=1}^{m_k}f_{k}(c_i^k)}}     \\
&\text{subject to}~~~~~~\sum_{k=1}^{K} \sum_{i=1}^{m_k} c_i^k\leq C,\\
&~~~~~~~~~~~~~~~~~~~c_i^{k}\geq 0, 1\leq k \leq K, 1\leq i \leq m_k.
\label{eq:Gaussianmultipletime}
\end{align}

As one example, suppose that $f_k(x)=\alpha_k^2 x$, where $\alpha_k$ is a constant, then one can obtain the solution to (\ref{eq:Gaussianmultipletime}) as
$$\sum_{i=1}^{m_k} c_i^k=  \frac{C  \|v_{k}\|\sigma_k/\alpha_k   }{\sum_{j=1}^{K} (\|v_{j}\|\sigma_j/\alpha_j)   }, ~~1\leq i \leq K, $$
and the corresponding biggest possible Cram\'{e}r-Rao  lower bound is given by
$$   \frac{1}{C}\left(\sum_{j=1}^{K} \|v_{j}\|\sigma_j/\alpha_j \right)^2.$$

As another example, let us consider the general case where we assume that
$$  \max_{c\geq 0} \frac{f_{k}(c)}{c}=(\alpha_{k}^{*})^{2} ,$$
where $c=c_{k}^{*}$ satisfies  $\frac{f_{k}(c^{*}_{k})}{c^{*}_{k}}= (\alpha_{k}^{*})^{2}$.

Under this assumption, one can show that asymptotically as $C \rightarrow  \infty$,  the corresponding biggest  possible Cram\'{e}r-Rao  lower bound  satisfies
$$ (1+o(1))\frac{1}{C}\left(\sum_{j=1}^{K} \|v_{j}\|\sigma_j/\alpha_{k}^{*} \right)^2, $$
$$C_{k}=  (1+o(1))\frac{C  \|v_{k}\|\sigma_k/\alpha_k^{*}   }{\sum_{j=1}^{K} (\|v_{j}\|\sigma_j/\alpha_j^{*})   }, ~~1\leq i \leq K,    $$
and $$m_{k}=(1+o(1))\frac{C  \|v_{k}\|\sigma_k/\alpha_k^{*}   }{c^{*}_{k}\sum_{j=1}^{K} (\|v_{j}\|\sigma_j/\alpha_j^{*})   }, ~~1\leq i \leq K. $$

We remark that when $T(\theta)$ is a linear function of $\theta$, and $\mu_k(\theta)$ is a linear function of $\theta$ for every $1\leq k \leq K$, then the optimal sampling strategy is universal for every possible $\theta \in \mathcal{R}^{K}$.
Namely, the optimizer does not need to have prior knowledge of the domain of $\theta$ in optimizing the sampling strategy.   We also remark that,  when $f_k(x)=\alpha_k^2 x$, for the optimal sensing strategy, the optimal number of samples
for each data source  can be of any positive number; however, when $f_k(x)$ is a general nonnegative increasing function in $x$, then the optimal sensing strategy also needs to determine the optimal number $m_{k}$ of samples for each data source $k$.


\section{Clean Sensing for Accurate Election Opinion Polling: Optimal Strategies under Distorted Data}
\label{sec:polling}

In this section, we consider applying the clean sensing framework to the problem of opinion polling for an (political) election, and explains one possible reason for  the inaccuracies of the polling for the 2016 presidential election through this framework.  We first introduce our mathematical model for the opinion polling.

\subsection{Mathematical Modeling of Election Opinion Polling from Heterogeneous Demographic Groups }

In an election, we assume that there are two candidates for the targeted position: candidate $A$ and candidate $B$, for the simplicity of our analysis. We however remark that our analysis can be easily extended to include more than 2 candidates. While our analysis can be extended to incorporate the cases of undecided voters or non-voting citizens, for simplicity of analysis, we assume that every citizen will vote, and every citizen will either vote for Candidate $A$ or vote for Candidate $B$.
We also assume that each citizen has made up their minds about what candidate he/she will vote for at the time of opinion polling. 

We consider $K$ demographic groups, and assume that $\theta_k$ fraction of people from the $k$-th demographic group eventually vote for candidate $A$ , where $1\leq k \leq K$ and $0\leq \theta_k \leq 1$. We assume that the population of each demographic group is large enough, such that a person randomly polled from the $k$-th demographic group eventually votes for candidate $A$ with probability $\theta_k$. Moreover, we assume that the eventual voting decisions of polled persons are independent of each other within a demographic group, and across different demographic groups.  We assume that the pollster randomly selects without replacement $m_{k}$ people to ask for their opinions. We use random variable $Z_i^k$ to represent how the $i$-th ($1\leq i \leq m_{k}$)  person polled from the demographic group $k$  ($1\leq k \leq K$) will eventually vote in the actual election: if the polled person votes for Candidate $A$, then $Z_i^k=1$; otherwise, $Z_i^k=0$.  As discussed above, we assume that $Z_{i}^{k}$ are independent random variables,  within a demographic group, and across different demographic groups

Moreover, when we sample the opinion of a randomly selected person from the $k$-th demographic group, we assume  that person will  always give a response of whether he/she will vote for Candidate $A$ or Candidate $B$.   We use random variable $X_i^k$ to represent the test response of the $i$-th polled person from the $k$-th demographic group: if the test result identify the $i$-th person from the $k$-th group will vote for Candidate $A$, then $X_i^k=1$; otherwise, $X_i^k=0$.   Suppose that we spend cost $c_{i}^{k}$ on testing the response of the $i$-th polled person from the $k$-th demographic group. For example, the pollster can take the low-cost path of simply asking that person for his/her opinions through a phone call or can take the high-cost path of taking the trouble to look at both his/her phone call response and his/her social media posts and other relevant information.

For each $k$, We let $\beta_{k}(c_{i}^{k})$ and $\gamma_{k}(c_{i}^{k})$ be two functions with $c_{i}^{k}$ as variables, and assume that they take values between $0$ and $1$.  We assume that conditioning on $Z_{i}^{k}=1$,  $X_{i}^{k}=1$ with probability $1-\beta_{k}(c_{i}^{k})$, and  $X_{i}^{k}=0$ with probability $\beta_{k}(c_{i}^{k})$, where $0\leq \beta_{k}(c_{i}^{k})\leq 1$; and that conditioning on $Z_{i}^{k}=0$,  $X_{i}^{k}=1$ with probability $\gamma_{k}(c_{i}^{k})$, and  $X_{i}^{k}=0$ with probability $1-\gamma_{k}(c_{i}^{k})$, where $0\leq \gamma_{k}\leq 1$. Namely, if a polled person eventually votes for Candidate $A$,  that person provides a different opinion when polled (tested), with probability  $\beta_{k}$; and  if a polled person eventually votes for Candidate $B$,  that person provides a different opinion when polled (tested), with probability  $\gamma_{k}(c_{i}^{k})$.

Moreover, we assume that
$$P(X_{1}^{1}, X_{2}^{1}, ..., X_{m_{1}}^{1}, X_{1}^{2}, X_{2}^{2},...,  X_{m_{k}}^{k} |  Z_{1}^{1}, Z_{2}^{1}, ..., Z_{m_{1}}^{1}, Z_{1}^{2}, Z_{2}^{2},...,  Z_{m_{k}}^{k} )= \prod_{k=1}^{K} \prod_{i=1}^{m_{k}} P(X_{i}^{k}| Z_{i}^{k}) .$$
Namely $X_{i}^{k}$'s are obtained from $Z_{i}^{k}$'s through a discrete memoryless channel in the language of information theory; or, in English, for a certain $i$ and $k$,  $X_{i}^{k}$ only depends on $Z_{i}^{k}$.
Thus $X_{i}^{k}=1$ with probability
$$ P_{1} (\theta_{k})=(1-\beta_{k}) \theta_{k}+\gamma_{k} (1-\theta_{k}) =(1-\beta_{k}-\gamma_{k}) \theta_{k} +\gamma_{k}, $$
and $X_{i}^{k}=0$ with probability
$$ P_{0}(\theta_{k})=\beta_{k} \theta_{k}+(1-\gamma_{k}) (1-\theta_{k}) =(1-\gamma_{k}) +  (\beta_{k}+\gamma_{k}-1)   \theta_{k},$$
where we abbreviate $\beta_{k}(c_{i}^{k})$ and $\gamma_{k}(c_{i}^{k})$ to $\beta_{k}$ and $\gamma_{k}$ respectively.

The goal of the estimator is to estimate $$\theta=\sum_{k=1}^{K} \alpha_{k} \theta_{k},$$ from the polled data $X_{i}^{k}$'s, where we assume that the $k$-th demographic group constitutes $\alpha_{k}$ fraction of the total voter population.
Then the  Cram\'{e}r-Rao bound for estimating $\theta$ is given by
$$ V=\sum_{k=1}^{K}  (\alpha_{k})^{2}  /V_{k} (\theta_{k}),     $$
where $V_{k} (\theta_{k})$ is the Fisher information for $\theta_{k}$, and $1/V_{k} (\theta_{k})$ is the Cram\'{e}r-Rao bound for estimating $\theta_{k}$ using the data from the $k$-th demographic group.

Then the Fisher information of $\theta_{k}$ provided by the $i$-th sample of the $k$-th demographic group is
\begin{align}
 I_{i}^{k}&=\frac{ (\frac{\partial P_{0} (\theta_{k})}{\partial \theta_{k} }  )^{2}   }{P_{0} (\theta_{k})  }+\frac{ (\frac{\partial P_{1} (\theta_{k})}{\partial \theta_{k} }  )^{2}   }{P_{1} (\theta_{k})  }      \\
        & =\frac{(1-\beta_{k} -\gamma_{k}  )^{2}}{ (1-\beta_{k}-\gamma_{k}) \theta_{k} +\gamma_{k}  }+\frac{(1-\beta_{k} -\gamma_{k}  )^{2}}{ (\beta_{k}+\gamma_{k}-1) \theta_{k} +1-\gamma_{k}  }\\
        &=\frac{(1-\beta_{k} -\gamma_{k}  )^{2}}{ \left((1-\beta_{k}-\gamma_{k}) \theta_{k} +\gamma_{k} \right)    \left(  1-[    (1-\beta_{k}-\gamma_{k}) \theta_{k} +\gamma_{k}   ]   \right)}
\end{align}
Then the Fisher information of $\theta_{k}$ provided by the $k$-th demographic group is given by
$$ V_{k} (\theta_{k})=\sum_{i=1}^{m_{k}} I_{i}^{k}.           $$

We note that, when $\beta_{k}=0$ and $\gamma_{k}=0$, $I_{i}^{k}=\frac{1}{\theta_{k}(1-\theta_{k})}$. This following lemma says the Fisher information achieves its biggest value when $\beta_{k}=0$ and $\gamma_{k}=0$.

\begin{lemma}
$I_{i}^{k} (\beta_{k}, \gamma_{k}) \leq \frac{1}{\theta_{k} (1-\theta_{k})},$  where $0\leq \beta_{k} \leq 1$ and $0\leq \gamma_{k} \leq 1$.
\end{lemma}

\begin{proof}
The claim is obvious when $\beta_{k}+\gamma_{k}=0$ or $\beta_{k}+\gamma_{k}=2$.  When $\beta_{k}+\gamma_{k}=1$, then $I_{i}^{k}=0 \leq \frac{1}{\theta_{k} (1-\theta_{k})}$.
Now let us consider the case where $0<\beta_{k}+\gamma_{k}<1$. Then we have
\begin{align}
 I_{i}^{k} (\beta_{k}, \gamma_{k})&=\frac{(1-\beta_{k} -\gamma_{k}  )^{2}}{ \left((1-\beta_{k}-\gamma_{k}) \theta_{k} +\gamma_{k} \right)    \left(  1-[    (1-\beta_{k}-\gamma_{k}) \theta_{k} +\gamma_{k}   ]   \right)}\\
             &=\frac{1}{          \left(\theta_{k} +\frac{\gamma_{k}}{1-\beta_{k}-\gamma_{k}} \right)    \left(  \frac{1-\gamma_{k}}{1-\beta_{k}-\gamma_{k}} -\theta_{k}   \right)}\\
             &\leq \frac{1}{\theta_{k} (1-\theta_{k})},
\end{align}
where in the last step we use the fact that $\frac{\gamma_{k}}{1-\beta_{k}-\gamma_{k}}$ is nonnegative, and $\frac{1-\gamma_{k}}{1-\beta_{k}-\gamma_{k}} \geq 1$.

We further consider the case $1<\beta_{k}+\gamma_{k}<2$. We first do a change of variable $\beta'=1-\beta_{k}$ and $\gamma'=1-\gamma_{k}$.  We thus have $0\leq \beta' \leq 1$, $0 \leq \gamma' \leq 1$ and $0<\beta'+\gamma'<1$.  We will show that $I_{i}^{k} (\beta_{k}, \gamma_{k})=I_{i}^{k} (\beta', \gamma')$, implying that $I_{i}^{k} (\beta_{k}, \gamma_{k}) \leq \frac{1}{\theta_{k} (1-\theta_{k})}$ when  $1<\beta_{k}+\gamma_{k}<2$. In fact

\begin{align}
 I_{i}^{k} (\beta_{k}, \gamma_{k})&=I_{i}^{k} (1-\beta', 1-\gamma')\\
&= \frac{(1-(1-\beta') -(1-\gamma')  )^{2}}{ \left((1-(1-\beta')-(1-\gamma')) \theta_{k} +(1-\gamma') \right)    \left(  1-[    (1-(1-\beta')-(1-\gamma')) \theta_{k} +(1-\gamma')   ]   \right)}\\
             &= \frac{(1-\beta' -\gamma'  )^{2}}{    \left(  1-[    (1-\beta'-\gamma') \theta_{k} +\gamma'   ]   \right)  \left((1-\beta'-\gamma') \theta_{k} +\gamma' \right)}   \\
             &=I_{i}^{k} (\beta', \gamma')\\
             &\leq \frac{1}{\theta_{k} (1-\theta_{k})}.
\end{align}
This concludes the proof of this lemma.
\end{proof}

\subsection{Optimal Polling Strategy for a Particular $(\theta_{1},..., \theta_{K})$}
In this subsection, we investigate finding the optimal polling strategy which minimizes the Cram\'{e}r-Rao bound of the mean-squared error of parameter estimation of $\theta$, for a particular parameter set $(\theta_{1},..., \theta_{K})$. One can also extend this work to minimize  the worst-case mean-squared error (minimax MSE) if  the parameter vector is known to be within a  certain set or to minimize the average-case mean-squared error if we have a prior knowledge of the distribution of the parameter vectors, as discussed above.  In this subsection, we only illustrate the results for a particular parameter set $(\theta_{1},..., \theta_{K})$, which is useful when the pollster knows that the parameter vector is within a  vicinity of that parameter vector.

We would like to decide the optimal number of polled people $m_{k}$ for each demographic group, and decide the the cost $c_{i}^{k}$ spent on polling the $i$-th person from the $k$-th demographic group.  The goal is to minimize the Cram\'{e}r-Rao bound of the mean-squared error of parameter estimation of $\theta$, under a total budget constraint $C$ for polling all the $K$ demographic groups.  So we have the following optimization problem formulation:
\begin{align}
&\min_{m_k, c_i^k}~~~~~~\sum_{k=1}^{K}  (\alpha_{k})^{2} / V_{k} (\theta_{k}) \\
&\text{subject to}~~~\sum_{k=1}^{K}\sum_{i=1}^{m_k} c_i^k \leq C,\\
&~~~~~~~~~~~~~~~~I_{i}^{k}(c_{i}^{k}) =\frac{(1-\beta_{k} (c_{i}^{k}) -\gamma_{k} (c_{i}^{k})  )^{2}}{ \left((1-\beta_{k} (c_{i}^{k})-\gamma_{k} (c_{i}^{k})) \theta_{k}  +\gamma_{k} (c_{i}^{k}) \right)    \left(  1-[    (1-\beta_{k}(c_{i}^{k})-\gamma_{k}(c_{i}^{k})) \theta_{k} +\gamma_{k}(c_{i}^{k})   ]   \right)},\\
&~~~~~~~~~~~~~~~~     V_{k} (\theta_{k}) =\sum_{i=1}^{m_{k}} I_{i}^{k}(c_{i}^{k}) ,  \\
&~~~~~~~~~~~~~~~m_k \in \mathbb{Z}_{\geq 0},\\
&~~~~~~~~~~~~~~~c_i^k\geq 0, 1\leq k \leq K, 1\leq i \leq m_k.
\label{eq:particularparameter}
\end{align}

We assume for every $\theta$, under a cost $c_{i}^{k}$, $I_{i}^{k}(c_{i}^{k})$ satisfies the following conditions:
\begin{itemize}
\item $I_{i}^{k}(c_{i}^{k})$ is a non-decreasing function in $c_{i}^{k}$ for $c_i^{k}\geq 0$;
\item $\frac{c_{k}}{I_{i}^{k}(c_{i}^{k})}$ is well defined for $c_{k}\geq 0$;
\item $\frac{c_{k}}{I_{i}^{k}(c_{i}^{k})}$ achieves its minimum value at a finite $c_k^{**} \geq 0$;
\item we denote the corresponding minimum value $\frac{c_{k}^{**}}{  I_{i}^{k} (c_{k}^{**})}$ as $p_k^*$.
\end{itemize}

We can now see that the clean sensing framework introduced in Section \ref{sec:CS} applies. In particular, specializing Theorem \ref{theorem:main} to the polling problem, we have the following theorem:

\begin{theorem}
When  $C\rightarrow \infty$,  the smallest possible lower bound $\sum_{k=1}^{K}  (\alpha_{k})^{2}  /V_{k} (\theta_{k})$ on the mean squared error of an unbiased $\theta$ is given by

$$\frac{1}{C} \left(  \sum_{k=1}^{{K}}\alpha_{k}\sqrt{p_k^*} \right)^2.$$

Moreover, when  $C\rightarrow \infty$, the optimal sensing strategy that achieves this smallest possible lower bound is given by
$$C_{k}^*= \frac{C   \alpha_{k}\sqrt{ p_k^*}}{\sum_{k=1}^{{K}}\alpha_{k}\sqrt{   p_k^*}}, $$
and
$$m_{k}^{*}=\frac{\frac{C   \alpha_{k} \sqrt{p_k^*}}{\sum_{k=1}^{{K}}\alpha_{k}\sqrt{p_k^*}}}{c_k^{**}}.$$
\end{theorem}

As we can see, for the most accurate polling in terms of minimizing the Cram\'{e}r-Rao bound, the total cost allocated for for a particular demographic group should be proportional to the population of that group, and proportional to the square root of the best ``cost-over-Fisher-information'' ratio for that group.  Namely, if the cost associated with  obtaining a certain amount of Fisher information from a person in a certain demographic group is high (which often implies polling data from that group is more distorted or more noisy), the total cost allocated for that group should also be high.  Because the number of persons polled from the $k$-th group is given by
$$m_{k}^{*}=\frac{\frac{C   \alpha_{k} \sqrt{\frac{c_{k}^{**}}{I_{i}^{k} (c_{k}^{**})}   }}{\sum_{k=1}^{{K}}\alpha_{k}\sqrt{p_k^*}}}{c_k^{**}}=\frac{C   \alpha_{k} \sqrt{\frac{1}{I_{i}^{k} (c_{k}^{**}) c_{k}^{**} }   }}{\sum_{k=1}^{{K}}\alpha_{k}\sqrt{p_k^*}},$$
the number of polled persons from the $k$-th demographic group should be inversely proportional to the square root of $I_{i}^{k} (c_{k}^{**}) c_{k}^{**}$, namely the ``Fisher-information-cost-product''  (at the best individual cost $c_{k}^{**}$) for the $k$-th group. This is in contrast to the common belief that the number of persons polled from a particular group is only proportional to the group's population.


%
%
%
%
%
%
%
%

\subsection{Comparisons with Polling using Plain Averaging of Polling Responses}
In this subsection, we will demonstrate that,  if clean sensing or other similar mechanisms are not used in guaranteeing the quality of data samples in election polling, the polling results can be quite inaccurate.  We show that,
when the qualities of individual data samples are not controlled to be good enough, more (big) data may not always be able to drive the polling error down to be small.

In particular, in this subsection, we investigate the polling error when plain averaging is used in estimating the parameter $\theta$.  In the plain averaging strategy,  the estimation $\hat{\theta}$ of parameter $\theta$ is given by
$$ \hat{\theta}=\sum_{k=1}^{K} \alpha_{i} \hat{\theta}_{k},    $$
where
$$\hat{\theta}_{k}=\frac{1}{m_{k}} \sum_{i=1}^{m_{k}}  X_{i}^{k}.$$

Then the mean-squared error of this estimation is given by
\begin{align}
E[( \theta-\hat{\theta} )^{2}]&=\left ( \sum_{k=1}^{K}  \alpha_{k} (\theta_{k}-\hat{\theta}_{k})  \right)^{2}         \\
&=\left ( E(\sum_{k=1}^{K}  \alpha_{k} (\theta_{k}-\hat{\theta}_{k}) )+   \sum_{k=1}^{K}  \alpha_{k} (\theta_{k}-\hat{\theta}_{k})-E(\sum_{k=1}^{K}  \alpha_{k} (\theta_{k}-\hat{\theta}_{k}) \right)^{2}   \\
&=\left ( E(\sum_{k=1}^{K}  \alpha_{k} (\theta_{k}-\hat{\theta}_{k}) )\right)^{2}+  \left( \sum_{k=1}^{K}  \alpha_{k} (\theta_{k}-\hat{\theta}_{k})-E(\sum_{k=1}^{K}  \alpha_{k} (\theta_{k}-\hat{\theta}_{k})) \right)^{2} \\
&=\left ( \sum_{k=1}^{K}  \alpha_{k} (\theta_{k}-E(\hat{\theta}_{k}) )\right)^{2}+  \sum_{k=1}^{K}   \alpha_{k}^2\left( \hat{\theta}_{k}-E( \hat{\theta}_{k}) \right)^{2}\\
&= \left ( \sum_{k=1}^{K}  \alpha_{k}  \left( \theta_{k}- \frac{1}{m_{k}}   \sum_{i=1}^{m_{k}}  \left( \theta_{k} +\gamma_{k}(c_{i}^{k}) -  (\beta_{k}(c_{i}^{k})  +\gamma_{k}   (c_{i}^{k}) ) \theta_{k}       \right)              \right )              \right)^2\\
&+ \sum_{k=1}^{K}   \alpha_{k}^2  \frac{ \sum_{i=1}^{m_k} \left[   (1-\beta_{k}(c_i^k)-\gamma_{k}(c_i^k)) \theta_{k} +\gamma_{k} (c_i^k) - \left(  (1-\beta_{k}(c_i^k)-\gamma_{k}(c_i^k)) \theta_{k} +\gamma_{k} (c_i^k)\right)^2        \right]     }{m_k^2}\\
&= \left ( \sum_{k=1}^{K}  \alpha_{k}  \left ( \frac{1}{m_{k}}  \sum_{i=1}^{m_{k}}  \left(\gamma_k(c_i^k) -(\beta_k (c_{i}^{k})+\gamma_k (c_{i}^{k}))\theta_k \right)    \right ) \right)^2\\
&+ \sum_{k=1}^{K}   \alpha_{k}^2  \frac{ \sum_{i=1}^{m_k} \left[   (1-\beta_{k}(c_i^k)-\gamma_{k}(c_i^k)) \theta_{k} +\gamma_{k} (c_i^k) - \left(  (1-\beta_{k}(c_i^k)-\gamma_{k}(c_i^k)) \theta_{k} +\gamma_{k} (c_i^k)\right)^2        \right]     }{m_k^2}\\
&\geq   \left ( \sum_{k=1}^{K}  \alpha_{k}  \left ( \frac{1}{m_{k}}  \sum_{i=1}^{m_{k}}  \left(\gamma_k(c_i^k) -(\beta_k (c_{i}^{k})+\gamma_k (c_{i}^{k}))\theta_k \right)    \right ) \right)^2 .
 \end{align}

If the cost spent on obtaining each sample from the $k$-th data source is always equal to $c_{k}$,  then $\beta_k (c_{k})$'s and $\gamma_k(c_{k})$'s are the same for the $k$-th data source, and we denote them by $\gamma_{k}$ and $\beta_{k}$. For such $\gamma_k$'s and $\beta_k$'s, we have
$$  \left ( \sum_{k=1}^{K}  \alpha_{k}  \left ( \frac{1}{m_{k}}  \sum_{i=1}^{m_{k}}  \left(\gamma_k(c_i^k) -(\beta_k (c_{i}^{k})+\gamma_k (c_{i}^{k}))\theta_k \right)    \right ) \right)^2 = \left ( \sum_{k=1}^{K}  \alpha_{k} (\gamma_k -(\beta_k+\gamma_k)\theta_k )\right)^2 .$$

If $\sum_{k=1}^{K}  \alpha_{k} (\gamma_k -(\beta_k+\gamma_k)\theta_k )\neq 0$, then the expected estimation error of $\theta$
will not go down to 0, even if the number of samples $m_k \rightarrow \infty$ for every $k$.  For example, let us consider two demographic groups, where $\alpha_1=0.5$, $\theta_1=0.5$, $\alpha_2=0.5$, and $\theta_2=0.5$. For the first group, $\beta_1=0.3$, and $\gamma_1=0$; and for the 2nd group, $\beta_2=0.1$, and $\gamma_2=0$. Then
$$E[( \theta-\hat{\theta} )^{2}] \geq  \left ( 0.5 (0-0.3\times 0.5)+0.5 (0-0.1\times 0.5)  \right)^{2}=0.01.$$

In fact, we can show that, under fixed  $\beta_{k}$ and $\gamma_{k}$, as the number of samples $m_{k} \rightarrow \infty$ for every $k$,  $\theta- \hat{\theta}$ will converge to $-\sum_{k=1}^{K}  \alpha_{k} (\gamma_k -(\beta_k+\gamma_k)\theta_k )$ almost surely.  For the example discussed above, this means that $\theta- \hat{\theta}$ will converge to $-\left ( 0.5 (0-0.3\times 0.5)+0.5 (0-0.1\times 0.5)  \right)=0.1$. We can see that when $|\sum_{k=1}^{K}  \alpha_{k} (\gamma_k -(\beta_k+\gamma_k)\theta_k )|$  is big, the estimation error using plain averaging is also big, leading to inaccurate polling results even if the number of polled people is large.

\subsection{Optimal Polling Strategy for Plain Averaging under Distorted Polling Responses}

In this subsection, we consider designing optimal polling strategy for parameter estimation using plain averaging, under distorted polling (test) responses.  We assume that we have a total budget of $D$ for polling $K$ demographic groups. We need to decide the optimal number of polled persons for each particular group, and the optimal cost to be spent on polling each person from each particular group.  Our goal is to minimize the mean-squared error of estimating $\theta$, when the estimator uses plain averaging.  We remark that the mean-squared error of the plain-averaging estimation comes from two parts: the deviation of the mean of polling data from the actual parameter $\theta$, and the variance of the random estimated parameter $\hat{\theta}$.  The optimal strategy needs to allocate the polling budget in a balanced way to make sure that a sufficient large cost is spent on obtaining each sample such that  the deviation of the mean of polling data from the actual parameter $\theta$ is small, while, at the same time, to make sure that  a sufficiently large number of persons are polled to reduce the variance of the random estimate $\hat{\theta}$.  In this subsection, we will derive such an optimal polling strategy.  In our derivations, we assume that the functions $\beta(c_{i}^{k})$'s and $\gamma(c_{i}^{k})$'s are explicitly known to the polling strategy designer.  However, we remark that we can extend  our derivations to the cases where the polling strategy designer does not explicitly know the functions $\beta(c_{i}^{k})$'s and $\gamma(c_{i}^{k})$'s.  For example, we can also extend our derivations to the cases where $\beta(c_{i}^{k})$'s and $\gamma(c_{i}^{k})$'s are random functions (but remain the same function for each sample from the same data source) and the designer knows only statistical information about these two functions but not the functions themselves (in fact, the results for these extensions are very similar to Theorem \ref{thm:plainaveraging} below).   The results derived in the subsection are also useful in the scenario where the polling strategy designer only provides the polling data $X_{i}^{k}$'s, but not information about functions $\beta(c_{i}^{k})$'s and $\gamma(c_{i}^{k})$'s,  to the estimator:  the polling strategy designer knows explicitly the functions $\beta(c_{i}^{k})$'s and $\gamma(c_{i}^{k})$'s, but the estimator does not explicitly know the functions $\beta(c_{i}^{k})$'s and $\gamma(c_{i}^{k})$'s.

Our main results in this subsection is summarized in the following theorem.
\begin{theorem}
Let us consider $K$ independent data sources, where $K$ is a positive integer. Let $\alpha_k$, $\theta_k$, $\beta_k$ and $\gamma_k$ be defined the same as above.  Let us assume that a fixed cost  $c_{k}$ is spent on obtaining each sample from the $k$-th data source, where $1\leq k \leq K$.  Suppose the total budget available for sampling is given by $D$.  We assume that for every possible values for $c_k$'s,  $\sum_{k=1}^{K}  \alpha_{k} (\gamma_k -(\beta_k+\gamma_k)\theta_k )\neq 0$. We assume that, only when $c_k \rightarrow \infty$ for every $k$, $\sum_{k=1}^{K}  \alpha_{k} (\gamma_k -(\beta_k+\gamma_k)\theta_k )\rightarrow 0$.  We further assume that, when $c\rightarrow \infty$,
$$f_k (c)=\gamma_k (c) -(\beta_k(c)+\gamma_k(c))\theta_k (c)=(1+o(1))\frac{b_k}{c},$$
where $b_k$ is a positive number depending only on $k$.

Then as $D \rightarrow \infty$, to minimize the mean-squared error of the plain averaging,
$$D_k= \frac{\alpha_k  [b_k(\theta_k-\theta_k^2 )]^{\frac{1}{3}} }{ \sum_{k=1}^{K} \left( \alpha_k  [b_k(\theta_k-\theta_k^2 )]^{\frac{1}{3}} \right)} D. $$
The optimal sampling cost for each sample from the $k$-th data source is given by
$$c_k=D^{\frac{1}{3}} b_k^{ \frac{2}{3} } (\theta_k-\theta_k^2)^{-\frac{1}{3}}, $$
and the optimal number of persons polled from the $k$-th group is given by
$$m_k=\frac{   D^{\frac{2}{3}} \alpha_k b_k^{-\frac{1}{3}} (\theta_k-\theta_k^2)^{\frac{2}{3}} }{\sum_{k=1}^{K}      \left ( \alpha_k b_k^{\frac{1}{3}} (\theta_k-\theta_k^2)^{\frac{1}{3}}  \right)     }.             $$

Under the assumption of plain averaging, and the assumption of fixed cost for each sample from the same data source,  to estimate the parameter $\theta=\sum_{k=1}^{K} \alpha_k \theta_k$,  the smallest possible mean-squared estimation error is given by

$$E[( \theta-\hat{\theta} )^{2}]=\frac{2\left( \sum_{k=1}^{K}  \alpha_k  b_k^{\frac{1}{3}} (\theta_k-\theta_k^2)^{\frac{1}{3}} \right)^2 }{D^{\frac{2}{3}}}     $$

\label{thm:plainaveraging}
\end{theorem}

\begin{proof}
Suppose that the budget allocated for sampling from each data source is $D_k$, and the fixed cost to obtain each sample from the $k$-th data source is $c_k$, then the mean-squared estimation error is given by
$$\left ( \sum_{k=1}^{K}  \alpha_{k} (\gamma_k -(\beta_k+\gamma_k)\theta_k )\right)^2
+ \sum_{k=1}^{K}   \alpha_{k}^2  \frac{  \left[   (1-\beta_{k}-\gamma_{k}) \theta_{k} +\gamma_{k}  - \left(  (1-\beta_{k}-\gamma_{k}) \theta_{k} +\gamma_{k} \right)^2        \right]     }{\frac{D_k}{c_k}},$$
where $\gamma_k$ and $\beta_k$ are functions of $c_k$.

From the 1st term of this expression for the mean-squared error, in order to make the estimation error go to $0$ as $D\rightarrow \infty$, we must have $c_k \rightarrow \infty$ for all $k$'s,  as $D\rightarrow \infty$. Thus $\beta_k \rightarrow 0$ and $\gamma_k \rightarrow 0$ as $D\rightarrow \infty$.  So when $D\rightarrow \infty$, we can write  the 2nd term of the expression for the mean-squared error between $ (1-\epsilon)\sum_{k=1}^{K}   \alpha_{k}^2  \frac{  \left[    \theta_{k}   -  \theta_{k}^2          \right]     }{\frac{D_k}{c_k}}$  and  $ (1+\epsilon)\sum_{k=1}^{K}   \alpha_{k}^2  \frac{  \left[    \theta_{k}   -  \theta_{k}^2          \right]     }{\frac{D_k}{c_k}}$, where $\epsilon>0$ is an arbitrarily small positive number. Thus, as $D \rightarrow \infty$, we can minimize the mean-squared estimation error given by
$$\left ( \sum_{k=1}^{K}  \alpha_{k} (\gamma_k -(\beta_k+\gamma_k)\theta_k )\right)^2 + \sum_{k=1}^{K}   \alpha_{k}^2  \frac{  \left[    \theta_{k}   -  \theta_{k}^2          \right]     }{\frac{D_k}{c_k}}.$$  This formula can be changed to
\begin{align}
(\sum_{k=1}^{K}  \alpha_k f_k(c_k) )^2 + \sum_{k=1}^{K} \frac{c_k G_k}{D_k},
\label{eq:plainaverage_objective}
\end{align}
where
$$f_k (c_k)=\gamma_k (c_k) -(\beta_k(c_k)+\gamma_k(c_k))\theta_k=\frac{b_k}{c_k},$$
and
$$G_k=\alpha_k^2 (\theta_k-\theta_k^2).$$

Given $D_k$'s, we can find the optimal $c_k$'s which minimize (\ref{eq:plainaverage_objective}) as follows:
\begin{align}
c_k= \sqrt{\frac{D_k b_k \alpha_k}{G_k}}  \sqrt[3]{\sum_{k=1}^{K} \sqrt{ \frac{{\alpha_k b_k G_k}}{{D_k}}     }    },
\label{eq:optimalc}
\end{align}

Plugging $c_k$ into (\ref{eq:plainaverage_objective}), we can simplify (\ref{eq:plainaverage_objective}) as
\begin{align}
2(\sum_{k=1}^{K} \sqrt{\alpha_k} \tau_k   )^{\frac{4}{3}}=  2\left (\sum_{k=1}^{K} \sqrt{\frac{b_k \alpha_k^3 (\theta_k-\theta_k^2) }{D_k }}  \right)^{\frac{4}{3}} ,
\label{eq:objinvolvingD}
\end{align}
where we used $\tau_k=\sqrt{\frac{b_kG_k}{D_k}}=\sqrt{\frac{b_k \alpha_k^2 (\theta_k-\theta_k^2)}{D_k}}$.

Now we minimize (\ref{eq:objinvolvingD}) over $D_k$'s subject to the constraint that $\sum_{k=1}^{K} D_k \leq D$, where $1\leq k \leq K$. This minimizing $D_k$'s can be calculated as
\begin{align}
D_k=\frac{\alpha_k b_k^{\frac{1}{3}} (\theta_k- \theta_k^2)^{\frac{1}{3}}         }{\sum_{k=1}^{K} \alpha_k b_k^{\frac{1}{3}} (\theta_k- \theta_k^2)^{\frac{1}{3}}   } D.
\label{eq:optimalD_k}
\end{align}
Plugging (\ref{eq:optimalD_k}) into (\ref{eq:optimalc}), we can obtain the optimal $c_k$ as follows:
\begin{align}
c_k=D^{\frac{1}{3}} b_k^{ \frac{2}{3} } (\theta_k-\theta_k^2)^{-\frac{1}{3}}.
\label{eq:optimalc_everythingexplicit}
\end{align}

The number of samples obtained from the $k$-th data source, denoted by $m_k$, is given by
$$m_k=\frac{D_k}{c_k}=   \frac{D^{\frac{2}{3}} \alpha_k b_k^{-\frac{1}{3}} (\theta_k-\theta_k^2)^{\frac{2}{3}}     }{  \sum_{k=1}^{K} \alpha_k b_k^{\frac{1}{3}} (\theta_k- \theta_k^2)^{\frac{1}{3}}         }.$$

Plugging (\ref{eq:optimalc_everythingexplicit}) and (\ref{eq:optimalD_k}) into (\ref{eq:plainaverage_objective}), we can obtain the smallest achievable mean-squared error through plain averaging is given by
$$E[( \theta-\hat{\theta} )^{2}]=\frac{2\left( \sum_{k=1}^{K}  \alpha_k  b_k^{\frac{1}{3}} (\theta_k-\theta_k^2)^{\frac{1}{3}} \right)^2 }{D^{\frac{2}{3}}} .$$
\end{proof}

As can be seen Theorem \ref{thm:plainaveraging}, to achieve the smallest possible mean-squared error,  the cost spent on obtaining each sample increases  as the total budget $D$ increases, in the order of $O(D^{\frac{1}{3}})$; and the number samples from a data source also increases as the total budget $D$ increases, in the order of $O(D^{\frac{2}{3}})$.  We would like to contrast this result with traditional polling without  accounting for distorted data, namely when $\beta_{k}=0$ and $\gamma_{k}=0$.  If there is no distorted data,  we do not have to let the cost spent on per sample grow to $\infty$ as $D \rightarrow \infty$.

\section{New Lower Bounds on the Mean-Squared Error of Parameter Estimation which can be tighter than the Cram\'{e}r-Rao bound, and the Chapman-Robbins bound}
\label{sec:lowerbound}
 Parameter estimation using observed data is a classical problem in signal processing, and it is also a very important problem in big data analytics.  In parameter estimation,  we are often interested in  obtaining lower and upper bounds on the mean-squared error of parameter estimators.  It is particularly important to obtain lower bounds on the mean-squared error (or other error metrics) of parameter estimation, which give fundamental limits on the performance of every possible parameter estimator or parameter estimators of a particular class (such as unbiased parameters).  These lower bounds are useful for determining how far a certain parameter estimation algorithm runs from the optimal performance.  It is well-known that the Cram\'{e}r-Rao bound (CRB), and the Chapman-Robbins bound \cite{Chapmanbound} give lower bounds on unbiased parameter estimators or parameter estimators with known bias functions.

   In this section,  we will derive new and tighter lower bounds on the mean-squared error of parameter estimators. While  we can extend our results to biased estimators or estimation of vector parameters,  we will focus on our results for unbiased estimators of scalar parameters.  We have derived three series of new lower bounds on the mean-squared error of parameter estimators: Lower Bounds Series A (ABS), Lower Bounds Series B (BBS),  and Lower Bounds Series C (CBS). Our newly derived lower  bounds on the MSE of unbiased estimators can be tighter than the well-known Cram\'{e}r-Rao bound (CRB), and the Chapman-Robbins bound.   Our newly derived lower bounds on the MSE of unbiased estimators are the optimal objective values of certain convex optimization problems.   We can solve these convex optimization problems by looking at their Lagrange dual problems, and obtain closed-form new lower  bounds on the MSE.  Interestingly, we have discovered that both the Cram\'{e}r-Rao bound (CRB) and the Chapman-Robbins bound are  the optimal objective values of the Lagrange dual problems to  special cases of these convex programs, and thus are special cases of the newly derived lower bounds in this paper.

Suppose that we would like to estimate a scalar parameter $\theta$, and assume that the observation is denoted by a random variable $X$.  We further assume that probability density function of $X$ is given by $p(x;\theta)$, where $\theta$ is a parameter.  Suppose that there is an unbiased estimator $g(X)$ such that, for every $\theta$,
    $$ E_{p(x;\theta)} (g(X))=\theta.$$
We would like to bound the mean-squared estimation error
$$E_{p(x;\theta)} \left( (g(X)-\theta)^{2}    \right)= \int p(x;\theta)  (g(x)-\theta)^2 \,dx,$$
even though our optimization framework to derive lower bounds on estimation errors can be further extended to include other metrics, such as $E_{p(x;\theta)} \left( |g(X)-\theta|^{3}    \right)$.

\subsection{New Lower Bounds Series A (ABS)}

We first notice that
$$  \int p(x;\theta)  (g(x)-\theta) \,dx= \int p(x;\theta)  g(x) \,dx -\int p(x;\theta)  \theta \,dx=0. $$
Moreover, we notice that, for every $\theta_{1}$, $\theta_{2}$, and $\theta$, we always have
$$  \int p(x;\theta_{1}) ( g(x)-\theta) \,dx= \int p(x;\theta_{1})  g(x) \,dx -\int p(x;\theta_{1})  \theta \,dx=\theta_{1}-\theta, $$
$$  \int p(x;\theta_{2}) ( g(x)-\theta) \,dx= \int p(x;\theta_{2})  g(x) \,dx -\int p(x;\theta_{2})  \theta \,dx=\theta_{2}-\theta.$$
Thus for every $\theta_{1}$ and $\theta_{2}$, we have
$$  \int ( p(x;\theta_{1})-p(x;\theta_{2}) ) ( g(x)-\theta) \,dx=  (\theta_{1}-\theta)-(\theta_{2}-\theta)=\theta_{1}-\theta_{2}. $$

Thus, we know the MSE of any unbiased estimator is  lower bounded by the optimal objective value of the following optimization problem:
\begin{align}
&\overset{\text{minimize}}{g(x)} ~~~~~~~~~~  \int p(x;\theta)  (g(x)-\theta)^2 \,dx   \\
&\text{subject to}~~~~~~\int p(x;\theta)  (g(x)-\theta) \,dx=0,\\
&~~~~~~~~~~~~~~~~~~~\int ( p(x;\theta_{1})-p(x;\theta_{2}) ) ( g(x)-\theta) \,dx=  \theta_{1}-\theta_{2}, \forall \theta_1,~\forall \theta_2
\label{eq:optimizationforMSE}
\end{align}

We first notice that this is a convex program in $g(x)$, which can have an infinite number of constraints. If the domain of $x$ is infinite dimensional, then this program is an infinite-dimensional convex optimization problem. By considering only $N$ pairs of $\theta_{1,i}$ and $\theta_{2,i}$, where $1\leq i \leq N$, we have the MSE of any unbiased estimator is  lower bounded by the optimal objective value of the following optimization problem with a finite number of constraints:

\begin{align}
&\overset{\text{minimize}}{g(x)} ~~~~~~~~~~  \int p(x;\theta)  (g(x)-\theta)^2 \,dx   \\
&\text{subject to}~~~~~~\int p(x;\theta)  (g(x)-\theta) \,dx=0,\\
&~~~~~~~~~~~~~~~~~~~\int ( p(x;\theta_{1,i})-p(x;\theta_{2,i}) ) ( g(x)-\theta) \,dx=  \theta_{1,i}-\theta_{2,i}, 1\leq i \leq N.
\label{eq:optimizationforMSEfinite}
\end{align}

We let $h(x)=g(x)-\theta$, $p(x)=p(x;\theta)$,  $q_{i}(x)=p(x;\theta_{1,i})-p(x;\theta_{2,i})$ and $t_{i}= \theta_{1,i}-\theta_{2,i}$. Then the optimization problem can be reformulated as:
\begin{align}
&\overset{\text{minimize}}{h(x)} ~~~~~~~~~~  \int p(x)  h^2(x) \,dx  \nonumber \\
&\text{subject to}~~~~~~\int p(x)  h(x) \,dx=0,\nonumber\\
&~~~~~~~~~~~~~~~~~~~\int q_{i}(x) h({x}) \,dx=  t_{i}, 1\leq i \leq N.
\label{eq:optimizationforMSEfinite_simplified}
\end{align}

The Lagrangian associated with the optimization problem (\ref{eq:optimizationforMSEfinite_simplified}) is given by
\begin{align}
 & \int p(x)  h^2(x) \,dx+u \int p(x)  h(x) \,dx+\sum_{i=1}^{N} \lambda_{i} (  \int q_{i}(x) h(x) \,dx-  t_{i} )\\
&~=\int \left (p(x)  (h^2(x) +u h(x))+ \sum_{i=1}^{N} \lambda_{i}    q_{i}(x) h({x}) \right)\,dx-  \sum_{i=1}^{N} \lambda_{i} t_{i}.
\end{align}

For a fixed $x$,
$$ \left (p(x)  (h^2(x) +u h(x))+ \sum_{i=1}^{N} \lambda_{i}    q_{i}(x) h({x}) \right)$$ has a minimum value

$$ -\frac{\left( p(x)u+ \sum_{i=1}^{N} \lambda_{i}    q_{i}(x)       \right)^{2}}{4 p(x)}    $$
when
$$ h(x)= -\frac{\left( p(x)u+ \sum_{i=1}^{N} \lambda_{i}    q_{i}(x) \right)  }{2p(x)}.$$

Thus the Lagrange dual function associated with the optimization problem (\ref{eq:optimizationforMSEfinite_simplified})  is given by
\begin{align}
 & \min_{h(x)} \int p(x)  h^2(x) \,dx+u \int p(x)  h(x) \,dx+\sum_{i=1}^{N} \lambda_{i} (  \int q_{i}(x) h(x) \,dx-  t_{i} )\\
&~=-\int \frac{\left( p(x)u+ \sum_{i=1}^{N} \lambda_{i}    q_{i}(x)       \right)^{2}}{4 p(x)}\, dx-  \sum_{i=1}^{N} \lambda_{i} t_{i}.
\end{align}

The Lagrange dual problem to the optimization problem (\ref{eq:optimizationforMSEfinite_simplified})  is thus given by
\begin{align}
  \max_{u, \lambda_{1}, ..., \lambda_{N}} -\int \frac{\left( p(x)u+ \sum_{i=1}^{N} \lambda_{i}    q_{i}(x)       \right)^{2}}{4 p(x)}\, dx-  \sum_{i=1}^{N} \lambda_{i} t_{i}.
  \label{eq: dualproblembasic }
\end{align}

Expanding $\left( p(x)u+ \sum_{i=1}^{N} \lambda_{i}    q_{i}(x)       \right)^{2}$ to $p^{2}(x) u^{2}+(\sum_{i=1}^{N}   \lambda_{i} q_{i} (x))^{2}+ 2u p(x) (\sum_{i=1}^{N}   \lambda_{i} q_{i} (x))$,  we can change the dual problem to

\begin{align}
  &\max_{u, \lambda_{1}, ..., \lambda_{N}} -\int \frac{  p^{2}(x) u^{2}+(\sum_{i=1}^{N}   \lambda_{i} q_{i} (x))^{2}+ 2u p(x) (\sum_{i=1}^{N}   \lambda_{i} q_{i} (x))   }{4 p(x)}\, dx-  \sum_{i=1}^{N} \lambda_{i} t_{i}. \nonumber\\
  &=\max_{u, \lambda_{1}, ..., \lambda_{N}} -\frac{u^{2}}{4}-\int   \frac{  \sum_{i=1}^{N}   \lambda_{i}^{2} q_{i}^{2} (x)+ 2 \sum_{i=1}^{N} \sum_{j=1, j\neq i}^{N}  (\lambda_{i} \lambda_{j} q_{i} (x) q_{j}(x))   }{4 p(x)} \,dx   \nonumber\\
  &  ~~ -\int \frac{ 2u p(x)\sum_{i=1}^{N} q_{i}(x)  }{4p(x)} \,dx -  \sum_{i=1}^{N} \lambda_{i} t_{i}
  \label{eq: dualproblemexpanded}
\end{align}

 Because $\int q_{i}(x)\,dx=\int (p(x;\theta_{1,i})-p(x;\theta_{2,i}))\,dx=1-1=0,$  (\ref{eq: dualproblemexpanded}) is equal to
\begin{align}
 &\max_{u, \lambda_{1}, ..., \lambda_{N}} -\frac{u^{2}}{4}-\int   \frac{  \sum_{i=1}^{N}   \lambda_{i}^{2} q_{i}^{2} (x)+ 2 \sum_{i=1}^{N} \sum_{j=1, j\neq i}^{N}  (\lambda_{i} \lambda_{j} q_{i} (x) q_{j}(x))   }{4 p(x)} \,dx -\sum_{i=1}^{N} \lambda_{i} t_{i}  \\
  &=\max_{\lambda_{1}, ..., \lambda_{N}} -\int   \frac{  \sum_{i=1}^{N}   \lambda_{i}^{2} q_{i}^{2} (x)+ 2 \sum_{i=1}^{N} \sum_{j=1, j\neq i}^{N}  (\lambda_{i} \lambda_{j} q_{i} (x) q_{j}(x))   }{4 p(x)} \,dx  -\sum_{i=1}^{N} \lambda_{i} t_{i},
  \label{eq: dualproblemfurthersimplified}
\end{align}
with the maximizing $u$ given by $u=0$.

  We consider an $N \times N$ matrix $B$ with its element  in the $i$-th row and $j$-th column as
    $$B_{ij}= \int \frac{q_{i}(x)q_{j} (x)}{p(x)} \,dx,$$
where $1\leq i \leq N$, and $1\leq j \leq N$.

We denote $\lambda=(\lambda_{1},\lambda_{2}, ..., \lambda_{N})^{T}$, and $t=(t_{1},t_{2}, ..., t_{N})^{T}$. Then   (\ref{eq: dualproblemfurthersimplified}) becomes
$$\max_{\lambda} -\frac{1}{4} \lambda^{T} B \lambda- \lambda^{T}t.       $$
Taking derivatives with respect to the vector $\lambda$, we have
$$-\frac{1}{2} B \lambda-t=0, $$
and thus the maximizing $\lambda$ is given by
$$  \lambda=-2B^{-1} t. $$
where we assume $B$ is full-rank. If $B$ is not invertible, we can replace $B^{-1}$ with $B^{\dagger}$, where the superscript ${\dagger}$ means the Moore-Penrose inverse. If $B$ is not invertible,  we assume that $t$ is within the column range of $B$ (otherwise the optimal value (\ref{eq: dualproblemfurthersimplified}) will be $\infty$, implying an unbiased estimator is impossible).

 The corresponding maximum objective value of the dual problem is
 \begin{align}
& -\frac{1}{4} (-2B^{-1}t)^{T} B (-2B^{-1}t)+2t^{T} B^{-1}t   \nonumber\\
&=-t^{T} B^{-1}t+2t^{T}B^{-1}t \nonumber\\
&=t^{T}B^{-1}t.
\end{align}

Since the optimization problem (\ref{eq:optimizationforMSEfinite}) has no more constraints than the optimization problem  (\ref{eq:optimizationforMSE}) , the optimal objective value of (\ref{eq:optimizationforMSE}) is no smaller than the optimal objective value of  (\ref{eq:optimizationforMSEfinite}).  By the weak duality,  the optimal objective value of (\ref{eq:optimizationforMSEfinite}) is bigger than or equal to $t^{T}B^{-1}t$ (if an unbiased estimator is feasible at all, we also have the strong duality, by the Slater's condition, meaning that  $t^{T}B^{-1}t$ is equal to the objective value of   (\ref{eq:optimizationforMSEfinite}) ). Thus the optimal objective value of  (\ref{eq:optimizationforMSE})  is at least $t^{T}B^{-1}t$.

Let us summarize the new lower bound on the MSE of an unbiased estimator. Let $\theta_{1, i}$ and $\theta_{2,i}$'s be any $N$ pair of values from the domain $\mathcal{D}$ of the parameter, where $1\leq i \leq N$ and $N$ is any positive integer.   We define
\begin{align*}
t=  \begin{bmatrix}
           \theta_{1,1}-\theta_{2,1} \\
           \theta_{1,2}-\theta_{2,2}  \\
           \theta_{1,3}-\theta_{2,3}  \\
           \vdots \\
            \theta_{1,i}-\theta_{2,i}  \\
           \vdots \\
           \theta_{1,N}-\theta_{2,N}  \\
         \end{bmatrix},~
~~\text{and}~~~    B  \in \mathbb{R}^{N \times N},
\end{align*}
where  $B$'s elements in its $i$-th row and $j$-th column is given by
$$ B_{i, j}=\int  \frac{ (p(x; \theta_{1, i})-p(x; \theta_{2, i})   )  (p(x; \theta_{1, j})-p(x; \theta_{2, j})   )    }    {p(x; \theta)}      \,dx.  $$

Then for any unbiased estimator $g(X)$ satisfies
$$   E_{p(x;\theta)} \left( (g(X)-\theta)^{2}    \right)\geq     t^{T}B^{-1}t .$$

\subsection{Special Cases of the New Lower Bounds Series A on the MSE of an unbiased estimator: Connections with the Cram\'{e}r-Rao bound, and the Chapman-Robbins bound}
In this subsection, we consider some special cases of the newly derived Lower Bounds Series A on the MSE of an unbiased estimator, and show how the classical  Cram\'{e}r-Rao bound and Chapman-Robbins bound are connected with the newly derived bounds in this paper.

We first consider the case where we only have two constraints in the optimization problem (\ref{eq:optimizationforMSE}). We pick two numbers $\theta_1$ and $\theta_2$ from the domain of the parameter. Then the optimal objective value of following optimization problem gives a lower bound on the MSE of an unbiased estimator.

\begin{align}
&\overset{\text{minimize}}{g(x)} ~~~~~~~~~~  \int p(x;\theta)  (g(x)-\theta)^2 \,dx   \\
&\text{subject to}~~~~~~\int p(x;\theta)  (g(x)-\theta) \,dx=0,\\
&~~~~~~~~~~~~~~~~~~~\int ( p(x;\theta_{1})-p(x;\theta_{2}) ) ( g(x)-\theta) \,dx=  \theta_{1}-\theta_{2}.
\label{eq:optimizationforMSEfiniteonly1constraint}
\end{align}

By our derivations above, the optimal objective value to (\ref{eq:optimizationforMSEfiniteonly1constraint}) is given by
\begin{align}
E_{p(X;\theta)} \left( (g(X)-\theta)^{2}    \right) \geq \frac{(\theta_1-\theta_2)^2}{   \int \frac{ (p(x;\theta_1)- p(x;\theta_2))^2  }{p(x;\theta)} \, dx}.
\label{eq:resultsingleconstraint}
\end{align}

We consider several special cases of $\theta_1$ and $\theta_2$ as follows.\\

\noindent\textbf{Special Case a)}:\\
In this case, we take $\theta_1$=$\theta$, and $\theta_2=\theta+\delta$. Then as $\delta \rightarrow 0$,
$$ \frac{(\theta_1-\theta_2)^2}{   \int \frac{ (p(x;\theta_1)- p(x;\theta_2))^2  }{p(x;\theta)} \, dx} \rightarrow  \frac{\delta^2}{   \int \frac{ (p'(x;\theta) \delta)^2  }{p(x;\theta)} \, dx}= \frac{1}{   \int \frac{ (p'(x;\theta) )^2  }{p(x;\theta)} \, dx}, $$
where $p'(x;\theta)=\frac{d p(x;\theta)}{d \theta}$.
This gives the classical Cram\'{e}r-Rao bound. \\

\noindent\textbf{Special Case b)}:\\
In this case, we take $\theta_1$=$\theta$, and $\theta_2=\theta+\delta$. Then
$$ \frac{(\theta_1-\theta_2)^2}{   \int \frac{ (p(x;\theta_1)- p(x;\theta_2))^2  }{p(x;\theta)} \, dx} =  \frac{\delta^2}{   \int \frac{ (p(x;\theta)- p(x;\theta+\delta))^2  }{p(x;\theta)} \, dx}. $$
One can take the supremum of $ \frac{\delta^2}{   \int \frac{ (p(x;\theta)- p(x;\theta+\delta))^2  }{p(x;\theta)} \, dx}$ over  the set of $\delta$, and this gives the Chapman-Robbins bound. \\

\noindent\textbf{Special Case c)}:\\
In this case, we take a general $\theta_1$, and $\theta_2=\theta_1+\delta$. Then as $\delta \rightarrow 0$,
$$ \frac{(\theta_1-\theta_2)^2}{   \int \frac{ (p(x;\theta_1)- p(x;\theta_2))^2  }{p(x;\theta)} \, dx} \rightarrow  \frac{\delta^2}{   \int \frac{ (p'(x;\theta_1) \delta)^2  }{p(x;\theta)} \, dx}= \frac{1}{   \int \frac{ (p'(x;\theta_1) )^2  }{p(x;\theta)} \, dx}, $$
where $p'(x;\theta_1)=\frac{d p(x;\theta_1)}{d \theta_1}$.
This gives a new lower bound on the MSE of an unbiased estimator, which is different from both the Cram\'{e}r-Rao bound, and the Chapman-Robbins bound.  This new bound can be tighter than both the Cram\'{e}r-Rao bound, and the Chapman-Robbins bound\\

\noindent\textbf{Special Case d)}:\\
Because the lower bound (\ref{eq:resultsingleconstraint}) holds for arbitrary $\theta_1$ and $\theta_2$, we have
\begin{align}
E_{p(x;\theta)} \left( (g(X)-\theta)^{2}    \right) \geq   \max_{\theta_1, \theta_2}\frac{(\theta_1-\theta_2)^2}{   \int \frac{ (p(x;\theta_1)- p(x;\theta_2))^2  }{p(x;\theta)} \, dx}.
\label{eq:resultsingleconstraint_max}
\end{align}
This bound can be tighter than the Cram\'{e}r-Rao bound, the Chapman-Robbins bound, and the bound in Special Case c).

\subsubsection{A Simple Example for which the New Bounds are Tighter than the Cram\'{e}r-Rao bound, and the Chapman-Robbins bound}

Let us consider estimating a parameter $\theta \in [0, 1.5]$.  We assume that the observed random variable $X$ follows the Bernoulli distribution.  Moreover,
$P(X=1; \theta)=1-|1-\theta|$ and   $P(X=0; \theta)=|1-\theta|$.

Let us consider an unbiased estimation of $\theta$ from $X$, and let $\theta_{1}=0.25$. Then for the Cram\'{e}r-Rao bound, the lower bound on the MSE of estimating $\theta_{1}$ from $X$ is given by $\theta_{1}   (1-\theta_{1}) =\frac{3}{16} $.
The Chapman-Robbins bound is achieved when $\theta_{2}=1.5=\theta_{1}+\delta$, so $\delta=1.5-\theta_{1}=1.25$. Thus the  Chapman-Robbins bound on the MSE of an unbiased parameter estimator  is given by
$$   \frac{  (1.5-\theta_{1})^{2}          }{        \frac{ (1-|1-1.5|-\theta_{1})^{2} }{\theta_{1}}   +\frac{ (|1-1.5|-(1-\theta_{1}))^{2} }{1-\theta_{1}}           }  =\frac{75}{16}.     $$

In contrast, in the newly derived Lower Bounds Series A, we can look at the Special Case d) and consider $\theta_{1}'=0.6$ and $\theta_{2}'=1.4$. One can see that

$$  \frac{(\theta_1'-\theta_2')^2}{   \int \frac{ (p(x;\theta_1')- p(x;\theta_2'))^2  }{p(x;\theta)} \, dx} =+\infty. $$

In fact, one can easily see that because $\theta_{1}'=0.6$ and $\theta_{2}'=1.4$ produce totally the same probability distributions for $X$, it is not possible at all to have an unbiased parameter estimator for this estimation task.  Thus the MSE of an unbiased estimator should indeed be $\infty$.  Thus the bound given by Lower Bounds Series A is the tightest, and tighter than Cram\'{e}r-Rao bound, and the Chapman-Robbins bound. One can of course give many other examples showing the Lower Bounds Series A give tighter bounds than the Cram\'{e}r-Rao bound, and the Chapman-Robbins bound.

\subsection{New Lower Bounds Series B (BBS)}

In this subsection, we will derive new lower bounds on the MSE of an unbiased estimator, and we term these new bounds derived in this subsection as Lower Bounds Series B (BBS).  We will use different convex optimization problems, and their Lagrange dual problems to derive these new bounds.

Because for every $\theta_{1}$ and $\theta_{2}$, we have
$$  \int ( p(x;\theta_{1})-p(x;\theta_{2}) ) ( g(x)-\theta) \,dx=  (\theta_{1}-\theta)-(\theta_{2}-\theta)=\theta_{1}-\theta_{2},$$
then for any function $f(\theta_{1}, \theta_{2})$,
\begin{align*}
&  \int \left(\int f(\theta_{1}, \theta_{2})\left ( p(x;\theta_{1})-p(x;\theta_{2}) \right ) \, d \theta_{1}  \, d\theta_{2} \right) ( g(x)-\theta) \,dx \\
&=\int f(\theta_{1}, \theta_{2})\left ( \int ( p(x;\theta_{1})-p(x;\theta_{2}) ) ( g(x)-\theta) \,dx \right) \, d \theta_{1}  \, d\theta_{2}\\
&=   \int f(\theta_{1}, \theta_{2} ) (\theta_{1}-\theta_{2}) \, d \theta_{1}  \, d\theta_{2}.
\end{align*}

Let us take $N$ functions $f_{i}(\theta_{1}, \theta_{2})$, $1\leq i \leq N$, denote
$$\int f_{i}(\theta_{1}, \theta_{2})\left ( p(x;\theta_{1})-p(x;\theta_{2}) \right ) \, d \theta_{1}  \, d\theta_{2}= q_{i} (x),     $$
and denote
$$\int f_{i}(\theta_{1}, \theta_{2} ) (\theta_{1}-\theta_{2}) \, d \theta_{1}  \, d\theta_{2}= t_{i}.$$

Thus, we know the MSE of any unbiased estimator is  lower bounded by the optimal objective value of the following optimization problem:

\begin{align}
&\overset{\text{minimize}}{h(x)} ~~~~~~~~~~  \int p(x)  h^2(x) \,dx   \\
&\text{subject to}~~~~~~\int p(x)  h(x) \,dx=0,\\
&~~~~~~~~~~~~~~~~~~~\int q_{i}(x) h({x}) \,dx=  t_{i}, 1\leq i \leq N.
\label{eq:optimizationforMSEfinite_simplified_generalized}
\end{align}

We notice that (\ref{eq:optimizationforMSEfinite_simplified_generalized}) has the same format as the optimization problem (\ref{eq:optimizationforMSEfinite_simplified}). Moreover, we  still have
\begin{align*}
 \int q_i(x) \,dx &= \int \left(\int f_{i}(\theta_{1}, \theta_{2})\left ( p(x;\theta_{1})-p(x;\theta_{2}) \right ) \, d \theta_{1}  \, d\theta_{2}\right ) \, dx \\
 &= \int f_{i}(\theta_{1}, \theta_{2}) \left (\int  \left ( p(x;\theta_{1})-p(x;\theta_{2}) \right )  \, dx \right) \, d \theta_{1}  \, d\theta_{2}\\
 &= \int f_{i}(\theta_{1}, \theta_{2}) 0 \, d \theta_{1}  \, d\theta_{2}\\
 &=0.
\end{align*}

Thus we conclude that (\ref{eq:optimizationforMSEfinite_simplified_generalized}) has the same solution as the optimization problem (\ref{eq:optimizationforMSEfinite_simplified}), except that the values of $q_i(x)$'s and $t_i$'s are different. Again we consider an $N \times N$ matrix $B$ with its element  in the $i$-th row and $j$-th column as
    $$B_{ij}= \int \frac{q_{i}(x)q_{j} (x)}{p(x)} \,dx,$$
where $1\leq i \leq N$, and $1\leq j \leq N$.  We denote $\lambda=(\lambda_{1},\lambda_{2}, ..., \lambda_{N})^{T}$, and $t=(t_{1},t_{2}, ..., t_{N})^{T}$.

Then the variance of any unbiased estimator is always lower bounded by $t^{T}B^{-1}t$. However, we stress that these bounds can be tighter or at least as tight as the Lower Bounds Series A, depending no the choices of functions $f_{i} (\theta_{1}, \theta_{2})$'s.


\subsection{New Lower Bounds Series C (CBS)}
In this subsection, we will derive new lower bounds on the MSE of an unbiased estimator, and we term these new bounds derived in this subsection as Lower Bounds Series C (CBS).  We will use more general convex optimization problems, for which the integrals of $q_i(x)$'s over $x$ may not necessarily be 0,  and their Lagrange dual problems to derive these new bounds.

We notice that, for every $\theta_{1}$, we have
$$  \int  p(x;\theta_{1})  ( g(x)-\theta) \,dx=  \theta_{1}-\theta.$$
Then for any function $f(\theta_{1})$, we have
\begin{align*}
&  \int \left(\int f(\theta_{1}) p(x;\theta_{1})  \, d \theta_{1}   \right) ( g(x)-\theta) \,dx \\
&=\int f(\theta_{1})\left ( \int  p(x;\theta_{1})  ( g(x)-\theta) \,dx \right) \, d \theta_{1}  \\
&=   \int f(\theta_{1} ) (\theta_{1}-\theta) \, d \theta_{1}.
\end{align*}

Let us take $N$ functions $f_{i}(\theta_{1})$, $1\leq i \leq N$, denote
$$\int f_{i}(\theta_{1}) p(x;\theta_{1})  \, d \theta_{1} = q_{i} (x),     $$
and denote
$$  \int f_{i}(\theta_{1} ) (\theta_{1}-\theta) \, d \theta_{1}= t_{i}.$$

Thus, we know the MSE of any unbiased estimator is  lower bounded by the optimal objective value of the following optimization problem:

\begin{align}
&\overset{\text{minimize}}{h(x)} ~~~~~~~~~~  \int p(x)  h^2(x) \,dx   \\
&\text{subject to}~~~~~~\int p(x)  h(x) \,dx=0,\\
&~~~~~~~~~~~~~~~~~~~\int q_{i}(x) h({x}) \,dx=  t_{i}, 1\leq i \leq N.
\label{eq:optimizationforMSEfinite_simplified_generalizedseriesC}
\end{align}

Note that this problem can be solved in the same way as we have solved (\ref{eq:optimizationforMSEfinite_simplified}).  We recall that (\ref{eq: dualproblemexpanded}) is equal to
\begin{align*}
  &\max_{u, \lambda_{1}, ..., \lambda_{N}} -\int \frac{  p^{2}(x) u^{2}+(\sum_{i=1}^{N}   \lambda_{i} q_{i} (x))^{2}+ 2u p(x) (\sum_{i=1}^{N}   \lambda_{i} q_{i} (x))   }{4 p(x)}\, dx-  \sum_{i=1}^{N} \lambda_{i} t_{i}.\\
  &=\max_{u, \lambda_{1}, ..., \lambda_{N}} -\frac{u^{2}}{4}-\int   \frac{  \sum_{i=1}^{N}   \lambda_{i}^{2} q_{i}^{2} (x)+ 2 \sum_{i=1}^{N} \sum_{j=1, j\neq i}^{N}  (\lambda_{i} \lambda_{j} q_{i} (x) q_{j}(x))   }{4 p(x)} \,dx  \\
  &  ~~ -\int \frac{ 2u p(x)\sum_{i=1}^{N} q_{i}(x)  }{4p(x)} \,dx -  \sum_{i=1}^{N} \lambda_{i} t_{i}
  \label{eq: dualproblemexpanded_general}
\end{align*}

Now we note  $\int q_i(x) \, dx$ may not necessarily be 0.  So the maximizing $u$ is given by the number $u$ which maximizes
$$-\frac{u^2}{4}-\frac{u}{2} \left(\sum_{i=1}^{N}  \int q_i(x) \, dx     \right).        $$
Setting the derivative to $0$, we can obtain the maximizing $u$ is given by
$$u= -\left(\sum_{i=1}^{N}  \int q_i(x) \, dx     \right).$$
Thus the maximum value of $-\frac{u^2}{4}-\frac{u}{2} \left(\sum_{i=1}^{N}  \int q_i(x) \, dx     \right)$ is given by
$$ \frac{1}{4} \left(\sum_{i=1}^{N}  \int q_i(x) \, dx     \right)^2.$$

In summary, recall that $f_{i}(\theta_{1})$, $1\leq i \leq N$ are $N$ functions,
$$\int f_{i}(\theta_{1}) p(x;\theta_{1})  \, d \theta_{1} = q_{i} (x),     $$
and
$$\int f_{i}(\theta_{1} ) (\theta_{1}-\theta) \, d \theta_{1}= t_{i}.$$

We define
\begin{align*}
t=  \begin{bmatrix}
           \int f_{1}(\theta_{1} ) (\theta_{1}-\theta) \, d \theta_{1}\\
           \int f_{2}(\theta_{1} ) (\theta_{1}-\theta) \, d \theta_{1}  \\
           \vdots \\
            \int f_{i}(\theta_{1} ) (\theta_{1}-\theta) \, d \theta_{1}  \\
           \vdots \\
           \int f_{N}(\theta_{1} ) (\theta_{1}-\theta) \, d \theta_{1}  \\
         \end{bmatrix},~
    B  \in \mathbb{R}^{N \times N},
\end{align*}
where $B$'s element in its $i$-th row and $j$-th column is given by
$$ B_{i, j}=\int  \frac{ (\int f_{i}(\theta_{1}) p(x;\theta_{1})  \, d \theta_{1} )  (\int f_{j}(\theta_{1}) p(x;\theta_{1})  \, d \theta_{1})    }    {p(x; \theta)}      \,dx.  $$

Then any unbiased estimator $g(X)$ satisfies

$$   E_{p(x;\theta)} \left( (g(X)-\theta)^{2}    \right)\geq     t^{T}B^{-1}t+  \frac{ \left( \sum_{i=1}^{N} \int (\int f_{i}(\theta_{1}) p(x;\theta_{1})  \, d \theta_{1}) \, dx           \right)^2         }      {4} .$$

This newly derived Lower Bounds Series C can be tighter than the Lower Bounds Series A and Lower Bounds Series B, which can in turn be tighter than the Cram\'{e}r-Rao bound, and the Chapman-Robbins bound.

\section{Acknowledgment}
Weiyu Xu would like to thank H. C.  for conversations about presidential elections partially inspiring Xu's interest in this topic, and for helping with  cleaning and typesetting the solution to one of the Lagrange dual problems in LaTex.   W. Xu is also thankful to California Institute of Technology Alumni Association for sending him regularly hard copies of its magazine \emph{The Caltech Alumni Association Annual}, which introduced to W. Xu  the bug eating story  \cite{citmagazine} of Caltech alumnus Professor Sam Wang at Princeton University, and helped inspire this research.

\section{Appendix}



\subsection{Derivations of the Solution to (\ref{eq:optimizationworstcase and simplified1})}
\label{appendixKKT1}

We consider the following optimization problem, where $a_{1}$, ..., $a_{K}$ are positive constant numbers:
  $$\min_{b_1,b_2,...,b_K} f(b_1,b_2,...,b_K)=\frac{a_1}{b_1}+\frac{a_2}{b_2}+...+\frac{a_K}{b_K},$$
  subject to $$b_1+b_2+...+b_K\leq{C},$$
  and $$b_{k} \geq 0, 1\leq k \leq K.$$

  The Lagrangian of the optimization problem is given by
  \begin{align}
  L(b_1,b_2,...,b_K,\lambda, \lambda_{1}, ..., \lambda_{K})&=\frac{a_1}{b_1}+\frac{a_2}{b_2}+...+\frac{a_K}{b_K}+\lambda(b_1+b_2+...+b_K- C)\\
 & -\lambda_{1} b_{1}-\lambda_{2} b_{2}-...-\lambda_{K} b_{K},
  \end{align}
  where $\lambda$, $\lambda_{1}$, ..., and $\lambda_{K}$ are nonnegative numbers.

We look at the Karush-Kuhn-Tucker (KKT) conditions for the optimization problem above.
$$\left\{
 \begin{aligned}
  &\frac{\partial{L(b_1,b_2,...,b_K,\lambda, \lambda_{1}, ..., \lambda_{K)}}}{\partial{b_1}}=-\frac{a_1}{({b_1})^2}+\lambda-\lambda_{1}=0,\\
  &\frac{\partial{L(b_1,b_2,...,b_K,\lambda, \lambda_{1}, ..., \lambda_{K)}}}{\partial{b_2}}=-\frac{a_2}{({b_2})^2}+\lambda-\lambda_{2}=0,\\
  &...,\\
  &\frac{\partial{L(b_1,b_2,...,b_K,\lambda, \lambda_{1}, ..., \lambda_{K)}}}{\partial{b_K}}=-\frac{a_K}{({b_K})^2}+\lambda-\lambda_{K}=0,\\
  &\lambda(b_1+b_2+...+b_K-C)=0,\\
  & \lambda_{k} b_{k} =0, 1\leq k \leq K,\\
  &\lambda\geq 0, \\
  &\lambda_{k}\geq 0, 1\leq k \leq K.
 \end{aligned}
\right.$$
  Since $a_1,a_2,...,a_K>0$, we have:
$$\left\{
 \begin{aligned}
  &\frac{a_1}{(b_1)^2}=\frac{a_2}{(b_2)^2}=...=\frac{a_K}{(b_K)^2},\\
  &b_1+b_2+...+b_K=C,\\
  &b_{k}\geq 0, 1\leq k \leq K.
 \end{aligned}
\right.$$

Therefore, we must have
$$\left\{
 \begin{aligned}
  &b_1=\frac{C\sqrt{a_1}}{\sqrt{a_1}+\sqrt{a_2}+...+\sqrt{a_K}},\\
  &b_2=\frac{C\sqrt{a_2}}{\sqrt{a_1}+\sqrt{a_2}+...+\sqrt{a_K}},\\
  &...,\\
  &b_K=\frac{C\sqrt{a_K}}{\sqrt{a_1}+\sqrt{a_2}+...+\sqrt{a_K}},\\
  &\lambda=\frac{1}{C^{2}} \left(\sum_{k=1}^{K} \sqrt{a_{k}}\right)^{2},\\
  &\lambda_{k}=0, 1\leq k \leq K.
 \end{aligned}
\right.$$

One can check that, under these values for $b_{k}$'s, $\lambda_{k}$'s and $\lambda$, the KKT conditions are all satisfied.  Moreover, the optimal objective value is
$$  \frac{a_1}{b_1}+\frac{a_2}{b_2}+...+\frac{a_K}{b_K}      =\frac{1}{C}  (\sum_{k=1}^{K}  \sqrt{a_{k}} )^{2}.$$

For the optimization problem (\ref{eq:optimizationworstcase and simplified1}), we can set
$$ a_{k}= \left ( \frac{ \partial T(\theta)  }{ \partial \theta_{k}   } \right)^{2} p_{k}^{*}, 1\leq k \leq K , $$
and
$$ b_{k}=C_{k}, 1\leq k \leq K,$$
then we get
$$C_{k}^*= \frac{C   \frac{\partial T(\theta)}{\partial \theta_k} \sqrt{p_k^*}}{\sum_{k=1}^{{K}}\frac{\partial T(\theta)}{\partial \theta_k}\sqrt{p_k^*}}, $$
and, moreover, under the optimal $C_{k}^{*}$,  the optimal value of (\ref{eq:optimizationworstcase and simplified1})  is given by
$$\frac{1}{C} \left(  \sum_{k=1}^{{K}}\frac{\partial T(\theta)}{\partial \theta_k}\sqrt{p_k^*} \right)^2.$$

\subsection{Derivations of the Solution to (\ref{eq:optimizationworstcase with residue})  }
\label{appendixKKT2}

%


The Lagrangian of the optimization problem  (\ref{eq:optimizationworstcase with residue})  is given by
  \begin{align}
  L(C_1,C_2,...,C_K,\lambda, \lambda_{1}, ..., \lambda_{K}, \tau_{1}, ..., \tau_{K})&=     \sum_{k=1}^{K} \frac{ \left (\frac{\partial T(\theta)}{\partial \theta_k} \right )^{2}   }{ \frac{C_{k}}{p_{k}^{*}} -F^{k} (c_{k}^{**} )  }    +\lambda(C_1+C_2+...+C_K- C)\\
 &+\sum_{k=1}^{K}   \lambda_{k}  \left(\frac{C_{k}}{p_{k}^{*}} -F^{k} (c_{k}^{**})\right)-\sum_{k=1}^{K}   \tau_{k}  C_{k},
  \end{align}
  where $\lambda$, $\lambda_{1}$, $\tau_{1}$,  $\lambda_{2}$, $\tau_{2}$,  ..., $\lambda_{K}$ and $\tau_{K}$ are nonnegative numbers.

Now we look at the Karush-Kuhn-Tucker (KKT) conditions for the optimization problem (\ref{eq:optimizationworstcase with residue}).
$$\left\{
 \begin{aligned}
  &\frac{\partial{L(C_1,C_2,...,C_K,\lambda, \lambda_{1}, ..., \lambda_{K}, \tau_{1}, ..., \tau_{K})}}{\partial{C_1}}= \lambda+\frac{\lambda_{1}}{p_{1}^{*}}-\tau_{1} - \frac{ \left (\frac{\partial T(\theta)}{\partial \theta_1} \right )^{2}  \frac{1}{p_{1}^{*}}  }{\left( \frac{C_{1}}{p_{1}^{*}} -F^{1} (c_{1}^{**}, \theta_{1}) \right)^{2} }  =0,\\
  &\frac{\partial{L(C_1,C_2,...,C_K,\lambda, \lambda_{1}, ..., \lambda_{K}, \tau_{1}, ..., \tau_{K})}}{\partial{C_2}}= \lambda+\frac{\lambda_{2}}{p_{2}^{*}}-\tau_{2} -\frac{ \left (\frac{\partial T(\theta)}{\partial \theta_2} \right )^{2}  \frac{1}{p_{2}^{*}}  }{ \left(\frac{C_{2}}{p_{2}^{*}} -F^{2} (c_{2}^{**}, \theta_{2} ) \right)^{2} }        =0,\\
  &...,\\
  &\frac{\partial{L(C_1,C_2,...,C_K,\lambda, \lambda_{1}, ..., \lambda_{K}, \tau_{1}, ..., \tau_{K})}}{\partial{C_K}}=\lambda+\frac{\lambda_{K}}{p_{K}^{*}}-\tau_{K}- \frac{ \left (\frac{\partial T(\theta)}{\partial \theta_K} \right )^{2}  \frac{1}{p_{K}^{*}}  }{ \left(\frac{C_{K}}{p_{K}^{*}} -F^{K} (c_{K}^{**},\theta_{K} )  \right)^{2}}  =0,\\
  &\lambda(C_1+C_2+...+C_K-C)=0,\\
  & \tau_{k} C_{k} =0, 1\leq k \leq K,\\
  &\lambda\geq 0, \\
  & \lambda_{k}  \left(\frac{C_{k}}{p_{k}^{*}} -F^{k} (c_{k}^{**})\right)=0, 0 \leq k \leq K,\\
  &\lambda_{k}\geq 0, 1\leq k \leq K,\\
  &\tau_{k}\geq 0, 1\leq k \leq K.
 \end{aligned}
\right.
\label{eq:allKKT}
$$

We let $\tau_{k}=0$ and $\lambda_{k } =0$ for $1\leq k \leq K$. Then we have


$$ \lambda=\frac{\left ( \frac{\partial T(\theta)}{\partial \theta_k}  \right)^{2} }{\left(\frac{C_{k}}{p_{k}^{*}}-  F^{k} (c_{k}^{**}, \theta_{k})\right)^{2}}  \times \frac{1}{p_{k}^{*}} =\frac{\left ( \frac{\partial T(\theta)}{\partial \theta_k}  \right)^{2} p_{k}^{*}}{\left({C_{k}}-  F^{k} (c_{k}^{**}, \theta_{k}){p_{k}^{*}}\right)^{2}}.$$
By definition,  we have
$$  F^{k} (c_{k}^{**}, \theta_{k}){p_{k}^{*}}=c_{k}^{**},$$
implying that
$$ \lambda=\frac{\left ( \frac{\partial T(\theta)}{\partial \theta_k}  \right)^{2} p_{k}^{*}}{\left({C_{k}}-  c_{k}^{**}\right)^{2}} .$$

Thus if $\lambda \neq 0$ and $C_{k}\geq c_{k}^{**}$, we have
     $$C_{k}=\frac{\frac{\partial T(\theta)}{\partial \theta_k}\sqrt{p_{k}^{*}}}{\sqrt{\lambda}}+ c_{k}^{**}.$$
When $\lambda \neq 0$ and the KKT conditions are satisfied,  we have
$$ \sum_{k=1}^{K} C_{k}=C,$$
so
$$ \sum_{k=1}^{K} \left(\frac{\frac{\partial T(\theta)}{\partial \theta_k}\sqrt{p_{k}^{*}}}{\sqrt{\lambda}}+ c_{k}^{**}\right)=C.$$

Solving for $\lambda$ and $C_{k}$, we have
$$\lambda=  \left( \frac{ \sum_{k=1}^{K} \frac{\partial T(\theta)}{\partial \theta_k}\sqrt{ p_{k}^{*}}}{C-\sum_{k=1}^{K}c_{k}^{**} }    \right)^{2} , $$
and
$$C_{k}= \frac{\left(C-\sum_{k=1}^{K} c_{k}^{**}\right) \frac{\partial T(\theta)}{\partial \theta_k} \sqrt{ p_k^*}}{\sum_{k=1}^{{K}}\frac{\partial T(\theta)}{\partial \theta_k}\sqrt{p_k^*}}+c_{k}^{**}. $$

We can see that  when $C\geq \sum_{k=1}^{K} c_{k}^{**} $, the obtained $C_{k}$'s, $\lambda$, $\lambda_{k}$'s and $\tau_{k}$'s above satisfy all the KKT conditions listed in  (\ref{eq:allKKT}).  Thus these values will lead to the optimal value of  (\ref{eq:optimizationworstcase with residue}).  Thus, plugging the optimal $C_k$ into the objective function of (\ref{eq:optimizationworstcase with residue}), we have the optimal value of (\ref{eq:optimizationworstcase with residue}) is given by

$$\frac{1}{ C-\sum_{k=1}^{K}c_{k}^{**}  } \left(  \sum_{k=1}^{{K}}\frac{\partial T(\theta)}{\partial \theta_k}\sqrt{p_k^*} \right)^2.$$

\bibliography{Reference_NullSpaceCondition}
\end{document}